\documentclass[lettersize,journal,10pt]{IEEEtran}
\usepackage{amsmath,amsfonts}
\usepackage{cite}
\usepackage{algorithm}
\usepackage{algorithmicx}
\usepackage{algpseudocode} 
\usepackage{array}
\usepackage[caption=false,font=normalsize,labelfont=sf,textfont=sf]{subfig}
\usepackage{textcomp}
\usepackage{stfloats}
\usepackage{longtable}
\usepackage{url}
\usepackage{color}
\usepackage{verbatim}
\usepackage{graphicx}
\usepackage{textcomp}
\usepackage{booktabs} 
\usepackage{threeparttable}
\usepackage{multirow} 
\usepackage{amssymb}
\usepackage{booktabs}
\usepackage{bigstrut}
\usepackage{hyperref}
\usepackage{cleveref}
\usepackage{ulem}
\usepackage{cuted}

\newtheorem{theorem}{\textbf{Theorem}} 
\newtheorem{proof}{\textbf{Proof}} 
\begin{document}

\title{SMP-RCR: A Sparse Multipoint Moment Matching Method for RC Reduction} 

\author{Siyuan Yin, Yuncheng Xu, Lin Liu, Fan Yang,~\IEEEmembership{Member,~IEEE}, Xuan Zeng,~\IEEEmembership{Senior Member,~IEEE}, Chengtao An$^*$ and Yangfeng Su$^*$
\thanks{$^*\ $Corresponding authors: Yangfeng Su, Email: yfsu@fudan.edu.cn and Chengtao An, Email: ancht@empyrean.com.cn}
        \thanks{$^1\ $Siyuan Yin, Yuncheng Xu and Yangfeng Su are with School of Mathematical Sciences, Fudan University, China.}
        \thanks{$^2\ $Lin Liu and Chengtao An are with Empyrean, China.}
        \thanks{$^3\ $Fan Yang and Xuan Zeng are with State Key Lab of Integrated Chips and Systems, College of Integrated Circuits and Micro-Nano Electronics, Fudan University, China.}
}



\maketitle

\begin{abstract}
In post-layout circuit simulation, efficient model order reduction (MOR) for many-port resistor-capacitor (RC) circuits remains a crucial issue. The current mainstream MOR methods for such circuits include high-order moment matching methods and elimination methods. High-order moment matching methods—characterized by high accuracy, such as \textbf{PRIMA} and \textbf{TurboMOR}—tend to generate large dense reduced-order systems when the number of ports is large, which impairs the efficiency of MOR. Another common type of MOR method for many-port circuits is based on Gaussian elimination, with the \textbf{SIP} method as a representative. The main limitation of this method lies in the inadequate matching of high-order moments. In this paper, we propose a sparse multipoint moment matching method and present comprehensive theoretical analysis results regarding the multi-frequency high-order moment matching property. Meanwhile, to enhance the algorithm's efficiency, sparse control and deflation techniques are introduced to further optimize the algorithm. Numerical experiments demonstrated that, compared to \textbf{SIP}, the accuracy is improved by more than two orders of magnitude at high frequency points without adding many extra linear components. Compared to \textbf{TurboMOR} methods, our method achieves a speed improvement of more than twice while maintaining the same level of precision.
\end{abstract}


\section{Introduction}
\label{sec:intro}






In post-layout circuit simulation, a large number of parasitic resistors and capacitors are extracted after parameter extraction, which means we need to solve extremely large-scale nonlinear systems. This process is highly time-consuming. Model order reduction methods for nonlinear systems, as listed in \cite{MOR_Review_1, MOR_Review_2}, are used to construct approximate reduced nonlinear systems to speed up post-layout circuit simulation. However, methods such as proper orthogonal decomposition (\textbf{POD}) \cite{POD1, POD2, POD3} or balanced truncation (\textbf{TBR}) \cite{BT, TBR1, TBR2} are both computationally expensive when constructing the reduced model. To address this issue, nonlinear nodes are preserved, and the extracted parasitic resistor-capacitor (\textbf{RC}) networks are reduced using RC reduction (\textbf{RCR}) techniques \cite{RCR1} while ensuring the accuracy remains acceptable.

A kind of \textbf{RCR} method with high accuracy is based on \textbf{Krylov} subspace. This kind of method constructs an orthogonal basis by \textbf{Arnoldi} procedure  or \textbf{Lanczos} procedure and obtains the reduced system by projecting the original state space onto this subspace spanned by the orthogonal basis. \textbf{PRIMA}\cite{PRIMA} is based on the block \textbf{Arnoldi} procedure and \textbf{MPVL}\cite{MPVL} is based on \textbf{Lanczos} procedure, they can both match multiple moments on single point. Another kind of method with high accuracy is multipoint moment matching based methods based on rational \textbf{Lanczos} procedure\cite{IRKA} are proposed in \cite{TBLMOR} and \cite{ARBLMOR}. However, it is important to note that the reduced systems generated by these methods are dense. 

Another method, called \textbf{TurboMOR-RC}, which can match multiple moments, was proposed in \cite{TurboMOR}. The reduced systems obtained by these methods have a block tridiagonal structure. The size of the reduced systems generated by the above methods grows proportionally with the number of ports, leading to a decrease in the efficiency of model order reduction as the number of ports increases.

 A commonly used method for reducing the order of multi-port \textbf{RC} circuits is based on Gaussian elimination, with methods such as \textbf{SIP} \cite{SIP, SIP2, SparseRC}, PACT \cite{PACT}, and \textbf{TICER} \cite{TICER, TICER2, HDTICER} as representatives. These methods yield reduced-order systems whose dimension is independent of the number of ports. However, these methods currently face the challenge of low accuracy. From a theoretical perspective, these methods possess the property of matching up to the first two moments, but this may not always be sufficient in practical simulation processes involving high frequencies.
 
 Aggregation based methods such as \textbf{AMOR}\cite{AMOR} which apply spectral clustering \cite{Spectral_clustering} to aggregate internal nodes and \textbf{RC-in-RC-out} method\cite{RCinRCout} is also efficient for \textbf{RC} circuits with many ports. These methods generate sparse reduced systems without introducing fill-ins or negative elements. However, different from elimination based methods, these methods have no accuracy guarantee for moment match.

In this article, we 1) propose a framework to generate a reduced system with high accuracy for a many-port \textbf{RC} circuit based on elimination at multiple points; 2) present the multipoint moment matching theorem of the framework; 3) propose the implementation techniques including sparsity control and deflation to render the reduced system sparse. Numerical experiments show that, compared to \textbf{SIP}, our method delivers an accuracy improvement of more than two orders of magnitude at high-frequency points without introducing a large number of additional linear components. When benchmarked against \textbf{TurboMOR} methods, our approach reduces the number of non-zero elements by over 60\% and achieves a speedup of more than double, all while maintaining the same level of accuracy and retaining the block tridiagonal structure.

The rest of this paper is organized as follows. In Section 2, we describe the \textbf{MOR} problem for many-port \textbf{RC} circuits and review the \textbf{TurboMOR} algorithm. In Section 3, we introduce our multipoint moment matching model order reduction framework and provide corresponding moment matching theorem. Section 4 introduces implementation techniques for algorithm. Section 5 presents the numerical experiments, where our method is compared with commonly used \textbf{MOR} methods. In Section 6, we draw our conclusions, and in the Appendix, we provide some mathematical proofs and details.

\section{Preliminaries}
\label{sec:pre}
In this section, we mainly introduce the basic concepts of RC reduction and the fundamental notation. To facilitate comparison with the \textbf{TurboMOR} method, we also review the process of the \textbf{TurboMOR} algorithm for \textbf{RC} reduction.
\subsection{RC Reduction and Moment Matching}
The original system for model order reduction is derived from an \textbf{RC} circuit with $n$ nodes, including $p$ ports. By applying modified nodal analysis (\textbf{MNA})\cite{MNA}, the original system can be represented by the following differential algebraic equations:
\begin{equation}\label{original_system}
 \left\{
\begin{aligned}
&C\dot{x}(t) + Gx(t) = Bu(t)\\
&y(t)=B^\top x(t)\\
\end{aligned}
\right.,
\end{equation}
where matrices $G\in\mathbb{R}^{n\times n}$, $C\in\mathbb{R}^{n\times n}$ are symmetric, nonnegative definite, representing the conductance matrix and capacitance matrix, respectively. Vector $x\in\mathbb{R}^n$ contains nodal voltages. Matrix $B\in\mathbb{R}^{n\times p}$ denotes the index of the ports and there exists a permutation matrices $P,Q$ such that
\begin{equation*}
  PBQ = \begin{bmatrix}
         I_p \\
         0 
       \end{bmatrix}.
\end{equation*}
Vector $u\in\mathbb{R}^p$ and $y\in\mathbb{R}^p$ represent port currents and port voltages, respectively. Transfer function is used to describe the input-output relationship and the original transfer function is written as
\begin{equation}\label{H}
  H(s) = B^\top\left(G+sC\right)^{-1}B.
\end{equation}
It can be verified that $H(s)$ is invariant when we permutate the nodes.

\textbf{MOR} for original system (\ref{original_system}) involves finding a transformation matrix $V\in\mathbb{R}^{n\times r}$ and projection matrix $W\in\mathbb{R}^{n\times r}$, and $W^\top V = I_r$ is satisfied usually. Afterwards, the reduced model is generated by the approximation $x(t)\approx V\hat{x}(t)$ and the residuals of the state equations in system (\ref{original_system}) are orthogonal to $W$. To preserve the symmetry of reduced conductance matrix and reduced capacitance matrix, we set $W=V$. Therefore, the reduced system can be written as 
\begin{equation}
	\label{reduced_system}
	\left\{
	\begin{aligned}
     &\widehat{C}\dot{\hat{x}}(t) + \widehat{G}\hat{x}(t) = \widehat{B}u(t)\\
     &\widehat{y}(t)=\widehat{B}^\top \hat{x}(t)\\
	\end{aligned}
	\right.,
\end{equation}
where matrices $\widehat{G} = V^\top GV\in\mathbb{R}^{r\times r},\,\widehat{C}=V^\top CV\in\mathbb{R}^{r\times r}$, vector $\hat{x}\in\mathbb{R}^{r}$, matrix $\widehat{B}\in\mathbb{R}^{r\times p}$, and $r < n$. The transfer function of system (\ref{reduced_system}) is
\begin{equation}
    \label{H_hat}
	\hat{H}(s) = \widehat{B}^{\top}\left(\widehat{G}+s\widehat{C}\right)^{-1}\widehat{B}.
\end{equation}
In order to meet the requirements of circuit simulation, a satisfactory reduced-order system typically needs to satisfy the following properties:
\begin{itemize}
	\item Reduced system approximates original system well;
    \item Reduced system simulates faster than original system;
	\item Some properties of original systems such as passivity and symmetry are preserved;
	\item All the ports should be preserved, which means that matrix $\widehat{B}$ can be seen as the result of removing some zero rows from matrix $B$.
\end{itemize}

To measure the accuracy of the reduced-order system, we introduce the concept of the moment\cite{Moment_Match} of the transfer function. Let $A\equiv G+s_0C$ and expand the transfer function \cref{H} at $s_0\in\mathbb{R}_{+}$:
\begin{equation*}\label{Taylor_Expansion}
  H(s)=\sum\limits_{k=0}^{+\infty}(-1)^{k}B^\top\left(A^{-1}C\right)^kA^{-1}B(s-s_0)^k,
\end{equation*}
the moment of the $k$ th order at $s_0$ is defined as 
\begin{equation*}\label{Original_Moment}
  M_k(s_0)\equiv (-1)^{k}B^\top\left(A^{-1}C\right)^kA^{-1}B.
\end{equation*}
Similarly, we can define the moment of the reduced transfer function (\ref{H_hat}):
\begin{equation*}\label{Reduced_Moment}
  \widehat{M}_k(s_0)\equiv (-1)^{k}\widehat{B}^\top\left(\widehat{A}^{-1}\widehat{C}\right)^k\widehat{A}^{-1}\widehat{B},
\end{equation*}
where $\widehat{A} = \widehat{G}+s_0\widehat{C}$.
It has been proven in \cite{SIP} that the reduced system generated by \textbf{SIP} matches the first 2 moments at $s_0=0$ of original system.

\subsection{TurboMOR-RC}
The \textbf{TurboMOR} method further improves the accuracy based on the elimination method, while ensuring that the reduced-order system has a block tridiagonal structure.

Suppose the coefficient matrix of the original system (\ref{original_system}) has the following block partition:
\begin{equation*}
    G = \begin{bmatrix}
        G_{p} & G_{C}^\top\\
        G_{C} & G_{I}
    \end{bmatrix},\vspace{1em}C = \begin{bmatrix}
        C_{p} & C_{C}^\top\\
        C_{C} & C_{I}
    \end{bmatrix},\vspace{1em}B = \begin{bmatrix}
        I_p \\
        0
        \end{bmatrix}.
\end{equation*}
The first step of the algorithm is to completely decouple all internal nodes of the system through congruence transformation. Assume the \textbf{Cholesky} factor of $G_{I}$ is $K$, and the congruence transformation is formulated as
\begin{equation*}
    Q^{(1)} = \begin{bmatrix}
        I_p &\\
        -K^{-\top}K^{-1}G_{C} & K^{-\top}
    \end{bmatrix}.
\end{equation*} Then we have 
\begin{equation*}
    G^{(1)}\equiv Q^{(1)\top} G Q^{(1)}=\begin{bmatrix}
                                           G^{(1)}_{p} &  \\
                                            & I_{n-p} 
                                         \end{bmatrix};
\end{equation*}
\begin{equation*}
     C^{(1)}\equiv Q^{(1)\top} C Q^{(1)}=\begin{bmatrix}
                                           C^{(1)}_{p} & C^{(1)\top}_{C}\\
                                           C^{(1)}_{C} & C^{(1)}_{I}
                                         \end{bmatrix}.
\end{equation*}

After performing the congruence transformation, original system can be seen as the cascade of a system $\Sigma_1^{(1)}$:
\begin{equation*}
      \Sigma_1^{(1)}:\left\{\begin{aligned} &G^{(1)}_{p}x_p^{(1)}(s) + sC^{(1)}_{p}x_p^{(1)}(s) = Bu(s)+u_1^{(1)}(s)\\
      &y(s) = B^{\top} x_p^{(1)}(s)
      \end{aligned}\right.;
\end{equation*}
and a system $\Sigma_1^{(2)}$:
\begin{equation*}
      \Sigma_1^{(2)}:\left\{\begin{aligned} &x_I^{(1)}(s) + sC^{(1)}_{I}x_I^{(1)}(s) = -C_{C}^{(1)}u_2^{(1)}(s)\\
      &y^{(2)}(s) = -C_{C}^{\top} x_I^{(1)}(s)
      \end{aligned}\right.,
\end{equation*}
where $u_1^{(1)}(s)=sy^{(2)}(s)$ and $u_2^{(1)}(s) = sx_p^{(1)}(s)$.

The second step is to perform \textbf{QR} decomposition on the matrix $C_{C}^{(1)}$ in the system $\Sigma_2^{(1)}$:
\begin{equation*}
    C_{21}^{(1)} = \widetilde{Q}_2\begin{bmatrix}
        R^{(2)}\\
        0
    \end{bmatrix}.
\end{equation*} Then we can construct orthogonal transformation $Q^{(1)}$:
\begin{equation*}
    Q^{(2)}=\begin{bmatrix}
              I_p &  \\
                 & \widetilde{Q}^{(2)} 
            \end{bmatrix}.
  \end{equation*}The coefficient matrix of the original system after orthogonal transformation is
  \begin{equation*}
    G^{(2)}\equiv Q^{(2)\top} G^{(1)} Q^{(2)}=\begin{bmatrix}
                                           G^{(1)}_{p} &  &\\
                                            & I_p & \\
                                           & & I_{n-2p}
                                         \end{bmatrix},
  \end{equation*}
  \begin{equation*}
    C^{(2)}\equiv Q^{(2)\top} C^{(1)} Q^{(2)}=\begin{bmatrix}
                                           C^{(1)}_{p} & R^{(2)\top} &\\
                                           R^{(2)} & C^{(2)}_{p} & C^{(2)\top}_{C}\\
                                           &C^{(2)}_{C} & C^{(2)}_{I}
                                         \end{bmatrix}.
  \end{equation*}

  Repeating the process of the second step for another $r-2$ times and retaining only the first $rp$ nodes yields the final reduced-order system:
  \begin{figure*}
    \begin{equation}\label{Reduced_System_Turbo}
     \begin{bmatrix}
      G^{(1)}_{p} &  &  &  \\
       & I_{p}  &  \\
       &  & \ddots &  \\
       &  &  & I_{p}
    \end{bmatrix} \begin{bmatrix}
          x_p^{(1)}(s) \\
          x_p^{(2)}(s) \\
          \vdots\\
          x_p^{(r)}(s) 
        \end{bmatrix} + s\begin{bmatrix}
      C^{(1)}_{p} & R^{(2)\top} &  &  \\
      R^{(2)} & C^{(2)}_{p} & \ddots &  \\
       & \ddots & \ddots & R^{(r)\top} \\
       &  & R^{(r)} & C^{(r)}_{p}
    \end{bmatrix}\begin{bmatrix}
          x_p^{(1)}(s) \\
          x_p^{(2)}(s) \\
          \vdots\\
          x_p^{(r)}(s) 
        \end{bmatrix} = \begin{bmatrix}
          B \\
          0 \\
          \vdots\\
          0
    
        \end{bmatrix}u(s)\\
      \end{equation}
    \end{figure*}
The reduced-order system (\ref{Reduced_System_Turbo}) satisfies the following properties:
\begin{itemize}
	\item The first $2r$ moments are matched;
    \item $G$ is block diagonal and $C$ is block tridiagonal;
	\item The reduced order is $rp$;
	\item All the ports are preserved.
\end{itemize}
It can be seen that the order of the reduced-order system is proportional to the number of ports; therefore, for multi-port systems, the order reduction efficiency of this method will decrease.

\section{Multipoint Moment Matching Model Order Reduction Framework}
\label{sec:framework}
In this section, we introduce the framework of our proposed multipoint moment matching method. We transform the original system into an equivalent structure consisting of a series connection of multiple small systems and one large system via a series of congruence transformations and orthogonal transformations. Then we obtain our reduced-order system by neglecting the large system. 

Assume that all of the ports are reordered first and to make the input-output matrices can be represented by
    \begin{equation}\label{B}
      B = \begin{bmatrix}
            B^{(1)} \\
            0 
          \end{bmatrix}.
    \end{equation} Then the conductance matrix and the capacitance matrix are represented by
    \begin{equation}\label{GC}
      G=\begin{bmatrix}
          G_{p} & G_{C}^\top \\
          G_{C} & G_{I} 
        \end{bmatrix},\vspace{1em}C=\begin{bmatrix}
                                       C_{p} & C_{C}^\top \\
                                       C_{C} & C_{I} 
                                    \end{bmatrix},
    \end{equation}
    where $G_{p}$ and $C_{p}$ are $p$-by-$p$ submatrices of $G$ and $C$, which represents the connection relationship between ports. The frequency points are 
    \begin{equation}\label{S}
      S=[s_1,\,s_2,\cdots,s_m],
    \end{equation}
    where $s_1,\,s_2,\cdots,s_m$ are nonnegative real numbers.

We first introduce how to perform order reduction for the case of 2 frequency points, then generalize it to the case of multiple points, and present the corresponding moment matching theory.
\subsection{Matching two points}
    Firstly, we construct a congruence transformation to decouple ports and internal nodes at $s_1$. Define matrix $A$:
    \begin{equation*}\label{A}
      A\equiv G+s_1 C=\begin{bmatrix}
          A_{p} & A_{C} \\
          A_{C} & A_{I} 
        \end{bmatrix},
    \end{equation*}
     which is a positive symmetric matrix. Therefore, $A_{I}$ is nonsingular. Then original system  represented by \cref{original_system} is equivalent to the following system:
      \begin{equation}\label{Sigma_shift}
      \Sigma:\left\{\begin{aligned}  &A\begin{bmatrix}
                                                      x_p(s) \\
                                                      x_I(s) 
                                                    \end{bmatrix} + (s-s_1)C\begin{bmatrix}
                                                      x_p(s) \\
                                                      x_I(s) 
                                                    \end{bmatrix} = \begin{bmatrix}
                                                      B^{(1)} \\
                                                      0 
                                                    \end{bmatrix}u(s)\\
      &y(s) = B^{(1)\top} x_p(s)
      \end{aligned}\right..
    \end{equation}
We can construct the following congruence transformation:
     \begin{equation*}
       W_1\equiv \begin{bmatrix}
                 I_p &  \\
                 -A^{-1}_{p}A_{C} & I_{n-p} 
               \end{bmatrix}.
     \end{equation*}
     Perform the congruence transformation on the matrices $G$ and $C$, then we can obtain the following matrices:
     \begin{equation*}\label{G_1}
       \widehat{G}^{(1)} = W_1^\top G W_1= \begin{bmatrix}
           \widehat{G}^{(1)}_{p} &  \widehat{G}^{(1)\top}_{C} \\
           \widehat{G}^{(1)}_{C} &  \widehat{G}^{(1)}_{I} 
        \end{bmatrix};
     \end{equation*}   
     \begin{equation*}\label{C_1}
       \widehat{C}^{(1)} = W_1^\top C W_1= \begin{bmatrix}
           \widehat{C}^{(1)}_{p} &  \widehat{C}^{(1)\top}_{C} \\
           \widehat{C}^{(1)}_{C} &  \widehat{C}^{(1)}_{I} 
        \end{bmatrix},
     \end{equation*} 
     where 
     \begin{equation*}\label{G_11}
         \widehat{G}^{(1)}_{p} = G_{p} - G_{C}^\top A^{-1}_{I}A_{C} - A_{C}^\top A^{-1}_{I}G_{C} + A_{C}^\top A^{-1}_{I}G_{I}A^{-1}_{I}A_{C}; 
     \end{equation*}
     \begin{equation*}\label{C_11}
         \widehat{C}^{(1)}_{p} = C_{p} - C_{C}^\top A^{-1}_{I}A_{C} - A_{C}^\top A^{-1}_{I}C_{C} + A_{C}^\top A^{-1}_{I}C_{I}A^{-1}_{I}A_{C}; 
     \end{equation*}
     \begin{equation*}\label{G_21}
         \widehat{G}^{(1)}_{C} = G_{C} - G_{I}A^{-1}_{I}A_{C};
     \end{equation*}
     \begin{equation*}\label{C_21}
         \widehat{C}^{(1)}_{C} = C_{C} - C_{I}A^{-1}_{I}A_{C}. 
     \end{equation*}
     They satisfy the following two equations:
     \begin{equation}\label{A_111}
         \widehat{A}_{p}^{(1)}\equiv \widehat{G}^{(1)}_{p} + s_1\widehat{C}^{(1)}_{p} = A_{p} - A_{C}^\top A_{I}^{-1}A_{C};
     \end{equation}
     \begin{equation}\label{decouple}
         \widehat{G}^{(1)}_{C} + s_1\widehat{C}^{(1)}_{C}= 0.   
     \end{equation}
     Equation (\ref{decouple}) implies that the ports and internal nodes are decoupled at $s_1$ via the congruence transformation $W_1$. In elimination-based method such as \textbf{SIP}, all internal nodes are neglected. To achieve higher precision, we need to preserve additional internal nodes.

Secondly, we need to construct a orthogonal transformation to determine which internal nodes to preserve. We assume that $\widehat{C}^{(1)}_{C}$ has the following decomposition:
     \begin{equation*}\label{QR_C_21}
       \widehat{C}^{(1)}_{C} = Q^{(2)}\begin{bmatrix}
                                     B^{(2)} \\
                                     0 
                                   \end{bmatrix},
     \end{equation*} 
     where $Q^{(2)}$ is orthogonal matrix and $B^{(2)}$ is an $p_2$-by-$p$ matrix and $p_2 \leq p$. The specific decomposition method will be described in detail in \cref{sec:optimization}. Afterwards, we can define the orthogonal transformation $Q_2$:
     \begin{equation*}\label{Q_i}
       Q_2\equiv\begin{bmatrix}
                  I_p &  \\
                   & Q^{(2)} 
                \end{bmatrix}.
     \end{equation*}
 Define matrices
 \begin{equation*}
     G^{(2)}\equiv Q^{(2)\top} \widehat{G}_{I}^{(1)}Q^{(2)};
 \end{equation*}
 \begin{equation*}
     C^{(2)}\equiv Q^{(2)\top} \widehat{C}_{I}^{(1)}Q^{(2)};
 \end{equation*}
 \begin{equation*}
     B_2\equiv \begin{bmatrix}
                                     B^{(2)} \\
                                     0 
                                   \end{bmatrix}.
 \end{equation*}
When we perform the transformation $W_1$ and $Q_2$ on the system represented by \cref{Sigma_shift}, it is equivalent to  \cref{Sigma_shift_2},
\begin{figure*}
         \begin{equation}\label{Sigma_shift_2}
      \Sigma:\left\{\begin{aligned} & \begin{bmatrix}
                                        \widehat{A}^{(1)}_{p} & \\
                                          & G^{(2)} + s_1 C^{(2)}
                                      \end{bmatrix} \begin{bmatrix}
                                                      x^{(1)}_p(s) \\
                                                      x^{(1)}_I(s)\\ 
                                                    \end{bmatrix} + (s-s_1)\begin{bmatrix}
                                        \widehat{C}^{(1)}_{p} & B_2^\top \\
                                        B_2&C^{(2)} 
                                      \end{bmatrix}\begin{bmatrix}
                                                      x^{(1)}_p(s) \\
                                                      x^{(1)}_I(s)\\
                                                    \end{bmatrix} = \begin{bmatrix}
                                                      B^{(1)} \\
                                                      0 \\
                                                    \end{bmatrix}u(s)\\
      &y(s) = B^{(1)\top} x^{(1)}_p(s)
      \end{aligned}\right..
    \end{equation}
\end{figure*}
 where $\widehat{A}^{(1)}_{p}$ is defined as \cref{A_111}. Through equations $u^{(1)}(s)=(s-s_1)y^{(2)}(s)$ and $u_2(s) = (s-s_1)x^{(1)}_p(s)$, system (\ref{Sigma_shift_2}) can be seen as the cascade of a system $\Sigma_1^{(1)}$
    \begin{equation}\label{Sigma_1}
      \Sigma_1^{(1)}:\left\{\begin{aligned} &\widehat{G}^{(1)}_{p}x^{(1)}_p(s) + s\widehat{C}^{(1)}_{p}x^{(1)}_p(s) = B^{(1)}u(s)+u^{(1)}(s)\\
      &y(s) = B^{(1)\top} x^{(1)}_p(s)
      \end{aligned}\right.
    \end{equation}
    and a system $\Sigma_1^{(2)}$
    \begin{equation}\label{Sigma_2}
      \Sigma_1^{(2)}:\left\{\begin{aligned} &G^{(2)}x^{(1)}_I(s) + sC^{(2)}x^{(1)}_I(s) = -B_2u_2(s)\\
      &y_2(s) = -B_2^{\top} x^{(1)}_I(s)
      \end{aligned}\right..
    \end{equation}
Block the state vector $x_I^{(1)}$ as
    \begin{equation*}
        x_I^{(1)}(s)=\begin{bmatrix}
            x_p^{(2)}(s)\\
            x_I^{(2)}(s)
        \end{bmatrix},
    \end{equation*} where the length of $x_p^{(2)}$ is $p_2$. Then system $\Sigma_1^{(2)}$ can also be written as 
    \begin{equation*}
      \Sigma_1^{(2)}:\left\{\begin{aligned} &G^{(2)}\begin{bmatrix}
            x_p^{(2)}(s)\\
            x_I^{(2)}(s)
        \end{bmatrix} + sC^{(2)}\begin{bmatrix}
            x_p^{(2)}(s)\\
            x_I^{(2)}(s)
        \end{bmatrix} = -\begin{bmatrix}
            B^{(2)}\\
            0
        \end{bmatrix}u^{(2)}(s)\\
      &y_2(s) = -B^{(2)\top} x^{(2)}_p(s)
      \end{aligned}\right..
    \end{equation*} Although the nodes within $\Sigma_1^{(2)}$ are all internal nodes in the original system represented by \cref{Sigma_shift}, they can be divided into ports and internal nodes in $\Sigma_1^{(2)}$. We call the ports in $\Sigma_1^{(2)}$ linear ports because they are connected with the nodes within $\Sigma_1^{(1)}$ but neither input-output ports nor connected with nonlinear elements. All such linear ports need to be preserved.
    
    Thirdly, we decouple the linear ports and the internal nodes in $\Sigma_1^{(2)}$ at expansion point $s_2$. Define matrix $A^{(2)}$:
    \begin{equation*}
        A^{(2)}\equiv G^{(2)} +s_2C^{(2)}=\begin{bmatrix}
            A_p^{(2)} & A_C^{(2)\top} \\
            A_C^{(2)} & A_I^{(2)}
        \end{bmatrix}.
    \end{equation*} The congruence transformation is written as 
    \begin{equation*}
        W_2 = \begin{bmatrix}
            I_{p_2}&\\
            -\left(A_I^{(2)}\right)^{-1}A_C^{(2)}& I_{n-p_1-p_2}
        \end{bmatrix},
    \end{equation*} and we can obtain
    \begin{equation*}
        \widehat{G}^{(2)} = W_2^\top G^{(2)} W_2= \begin{bmatrix}
           \widehat{G}^{(2)}_{p} &  \widehat{G}^{(2)\top}_{C} \\
           \widehat{G}^{(2)}_{C} &  \widehat{G}^{(2)}_{I} 
        \end{bmatrix};
    \end{equation*}
    \begin{equation*}
        \widehat{C}^{(2)} = W_2^\top C^{(2)} W_2= \begin{bmatrix}
           \widehat{C}^{(2)}_{p} &  \widehat{C}^{(2)\top}_{C} \\
           \widehat{C}^{(2)}_{C} &  \widehat{C}^{(2)}_{I} 
        \end{bmatrix}.
    \end{equation*}  Perform the congruence transformation on $\Sigma_1^{(2)}$, and the original system can be seen as the cascade of two small systems
    \begin{equation*}
      \Sigma_1^{(1)}:\left\{\begin{aligned} &\widehat{G}^{(1)}_{p}x^{(1)}_p(s) + s\widehat{C}^{(1)}_{p}x^{(1)}_p(s) = B^{(1)}u(s)+u^{(1)}(s)\\
      &y(s) = B^{(1)\top} x^{(1)}_p(s)
      \end{aligned}\right.
    \end{equation*}
    \begin{equation*}
      \Sigma_2^{(1)}:\left\{\begin{aligned} &\widehat{G}^{(2)}_{p}x^{(2)}_p(s) + s\widehat{C}^{(2)}_{p}x^{(2)}_p(s) = B^{(2)}u_2(s)+u^{(2)}(s)\\
      &y^{(2)}(s) = B^{(2)\top} x^{(2)}_p(s)
      \end{aligned}\right.
    \end{equation*}
    and a large system
     \begin{equation*}
      \Sigma_2^{(2)}:\left\{\begin{aligned} &G_I^{(2)}x^{(2)}_I(s) + sC_I^{(2)}x^{(2)}_I(s) = -C_C^{(2)}u_3(s)\\
      &y^{(3)}(s) = -C_C^{(2)\top} x^{(2)}_I(s)
      \end{aligned}\right.,
    \end{equation*}
    where 
    \begin{equation*}
        u_k(s) = (s-s_{k-1})x_p^{(k-1)},\vspace{1em} k=2,3;
    \end{equation*}
    \begin{equation*}
        u^{(k)}(s) = (s-s_k)y^{(k+1)}(s),\vspace{1em}k=1,2.
    \end{equation*}
    
Finally, he reduced system for the two-point case can be obtained by neglecting $\Sigma_2^{(2)}$, which is represented by \cref{Reduced_system_2_point}. From the model order reduction process described above, the two frequency points are not required to be distinct. If the two frequency points satisfy $s_1=s_2$, a reduced system formulated as \cref{Reduced_system_2_point} can still be obtained; the only difference between these two scenarios lies in the moment matching theory, which will be discussed in the next subsection.

    \begin{figure*}[h]
        \begin{equation}\label{Reduced_system_2_point}
      \left\{\begin{aligned} &\begin{bmatrix}
          \widehat{G}_p^{(1)} & \widetilde{B}^{(2)\top}\\
          \widetilde{B}^{(2)} & \widehat{G}_p^{(1)}
      \end{bmatrix}\begin{bmatrix}
            x_p^{(1)}(s)\\
            x_p^{(2)}(s)
        \end{bmatrix} + s\begin{bmatrix}
          \widehat{C}_p^{(1)} & B^{(2)\top}\\
          B^{(2)} & \widehat{C}_p^{(1)}
      \end{bmatrix}\begin{bmatrix}
            x_p^{(1)}(s)\\
            x_p^{(2)}(s)
        \end{bmatrix} = \begin{bmatrix}
            B^{(1)}\\
            0
        \end{bmatrix}u(s)\\
      &y(s) = B^{(1)\top} x^{(1)}_p(s)
      \end{aligned}\right..
    \end{equation}
    \end{figure*}
    
\subsection{Matching more than two points}
In this subsection, the number of frequency points $m$ satisfies $m\geq 3$. When we obtain the 2 systems $\Sigma_1^{(1)}$ and $\Sigma_2^{(1)}$ as \cref{Sigma_1,Sigma_2}, via additional $m-1$ iterations shown in Appendix \ref{app:process}, the original system is transformed into the $m$ small systems
     \begin{equation}\label{Sigma_11}
      \Sigma_1^{(1)}:\left\{\begin{aligned} &\widehat{G}^{(1)}_{p}x^{(1)}_p(s) + s\widehat{C}^{(1)}_{p}x^{(1)}_p(s) = B^{(1)}u(s)+u^{(1)}(s)\\
      &y(s) = B^{(1)\top} x^{(1)}_p(s)
      \end{aligned}\right.;
    \end{equation}
    \begin{equation}\label{Sigma_21}
      \Sigma_2^{(1)}:\left\{\begin{aligned} &\widehat{G}^{(2)}_{p}x^{(2)}_p(s) + s\widehat{C}^{(2)}_{p}x^{(2)}_p(s) = -B^{(2)}u_2(s)+u^{(2)}(s)\\
      &y^{(2)}(s) = -B^{(2)\top} x^{(2)}_p(s)
      \end{aligned}\right.;
    \end{equation}
    \begin{equation*}
      \vdots
    \end{equation*}
    \begin{equation}\label{Sigma_m1}
      \Sigma_m^{(1)}:\left\{\begin{aligned} &\widehat{G}^{(m)}_{p}x^{(m)}_p(s) + s\widehat{C}^{(m)}_{p}x^{(m)}_p(s) =\\& -B^{(m)}u_m(s)+u^{(m)}(s)\\
      &y^{(m)}(s) = -B^{(m)\top} x^{(m)}_p(s)
      \end{aligned}\right.;
    \end{equation}
    and a large system
     \begin{equation}\label{Sigma_m2}
      \Sigma_m^{(2)}:\left\{\begin{aligned} &\widehat{G}^{(m)}_{I}x^{(m)}_I(s) + s\widehat{C}^{(m)}_{I}x^{(m)}_I(s) = -C_C^{(m)}u_{m+1}(s)\\
      &y^{(m+1)}(s) = B^{(m)\top} x^{(m)}_I(s)
      \end{aligned}\right.,
    \end{equation}
    where the input vectors and output vectors of the $m$ systems satisfy that
    \begin{equation*}
      u_k(s)=(s-s_{k-1})x_p^{(k-1)}(s),\vspace{1em}2\leq k\leq m+1;
    \end{equation*}
    \begin{equation*}
      u^{(k)}(s) = (s-s_k)y^{(k+1)}(s),\vspace{1em}1\leq k\leq m,
    \end{equation*}
    and \begin{equation*}
        B_m\equiv\begin{bmatrix}
            -B^{(m)}\\
            0
        \end{bmatrix}.
    \end{equation*}
    
    Finally, we neglect the large system $\Sigma_m^{(2)}$ and let \begin{equation*}
      \widetilde{B}^{(k)} = -s_{k-1}{B}^{(k)},\vspace{1em}2\leq k\leq m,
    \end{equation*} then the reduced system has been constructed as \cref{Reduced_System}.
    \begin{figure*}[h]
    \begin{equation}\label{Reduced_System}
      \widehat{\Sigma}:\left\{\begin{aligned} & \begin{bmatrix}
                                                  \widehat{G}^{(1)}_{p} & \widetilde{B}^{(2)\top} &  &  \\
                                                  \widetilde{B}^{(2)} & \widehat{G}^{(2)}_{p} & \ddots &  \\
                                                   & \ddots & \ddots & \widetilde{B}^{(m)\top} \\
                                                   &  & \widetilde{B}^{(m)} & \widehat{G}^{(m)}_{p}
                                                \end{bmatrix} \begin{bmatrix}
                                                      x^{(1)}_p(s) \\
                                                      x^{(2)}_p(s) \\
                                                      \vdots\\
                                                      x^{(m)}_p(s) 
                                                    \end{bmatrix} + s\begin{bmatrix}
                                                  \widehat{C}^{(1)}_{p} & B^{(2)\top} &  &  \\
                                                  B^{(2)} & \widehat{C}^{(2)}_{p} & \ddots &  \\
                                                   & \ddots & \ddots & B^{(m)\top} \\
                                                   &  & B^{(m)} & \widehat{C}^{(m)}_{p}
                                                \end{bmatrix}\begin{bmatrix}
                                                      x^{(1)}_p(s) \\
                                                      x^{(2)}_p(s) \\
                                                      \vdots\\
                                                      x^{(m)}_p(s) 
                                                    \end{bmatrix} = \begin{bmatrix}
                                                      B^{(1)} \\
                                                      0 \\
                                                      \vdots\\
                                                      0
                                                
                                                    \end{bmatrix}u(s)\\
      &y(s) = B^{(1)\top} x^{(1)}_p(s)
      \end{aligned}\right..
      \end{equation}
    \end{figure*}
   
\subsection{Moment Matching Analysis}
The property of moment matching is given by the following theorem. Consider the multiset $S = \{s_1, s_2, . . . , s_m\}$. We define $q(s_i)$ as the number of appearances of $s_i$ in $S$.
    \begin{theorem}\label{Moment_Match}
      The reduced system (\ref{Reduced_System}) matches the first $2q(s_i)$ moments of the original system (\ref{original_system}) at $s_i$, $i=1,2,\cdots,m$.
    \end{theorem}
    \begin{proof}
      The proof of the theorem first leverages the moment matching property of the elimination method to derive the moment matching property of the $m$-th system. Then, by leveraging the recursive relationship between two adjacent systems, the error of the transfer function of the first system is derived. The details are shown in Appendix \ref{app:convergence}.
    \end{proof}

    \Cref{Moment_Match} shows that when the frequency points satisfy
    \begin{equation*}
        s_1=s_2=\cdots=s_m,
    \end{equation*} the multipoint moment matching method has the same high-order moment matching properties as \textbf{PRIMA} and \textbf{TurboMOR}, thus ensuring its precision. However, this method faces the same problem as \textbf{TurboMOR}, as illustrated in \cref{Spy_G_hat_2,Spy_C_hat_2}. Although the system has a block tridiagonal structure, as the number of nodes and ports increases, each diagonal block becomes a dense matrix. This leads to a significant increase in the number of non-zero elements and reducing the efficiency of model order reduction.
    \begin{figure}[tb]
      \centering
      \includegraphics[width=\linewidth]{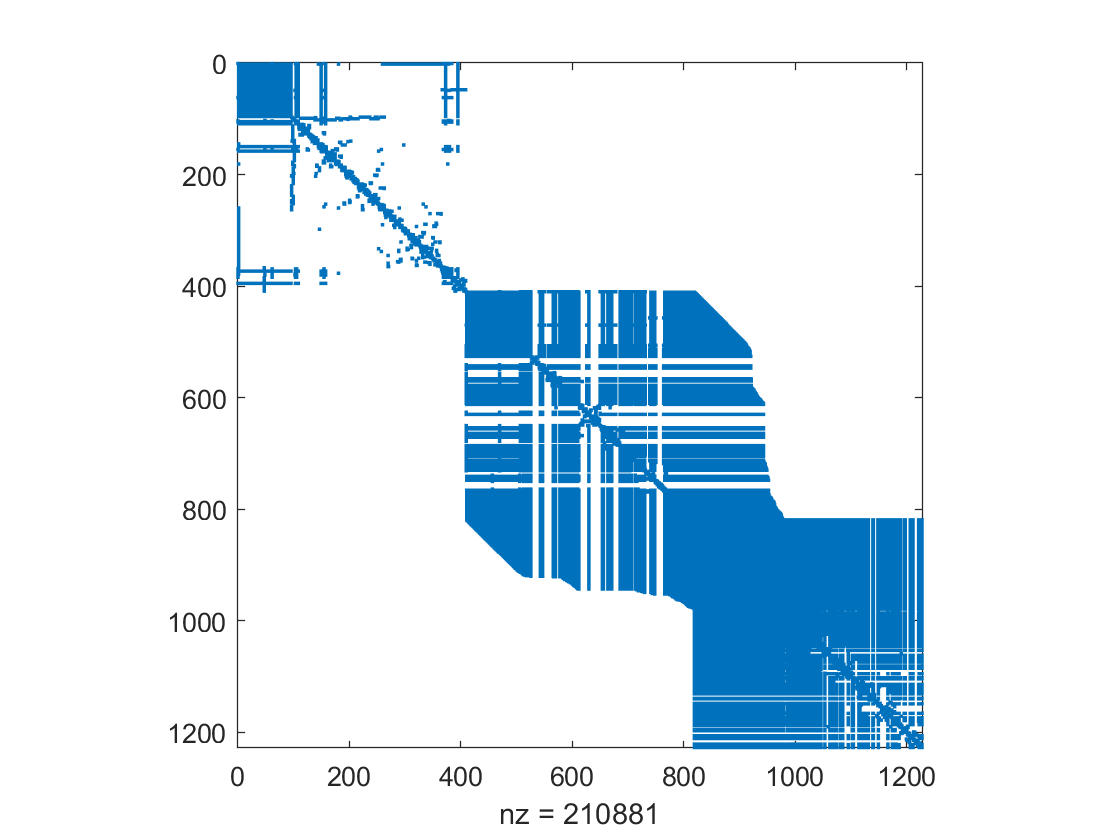}
      \caption{Nonzero Elements of the Reduced Conductance Matrix in ADC\_Net\_304}\label{Spy_G_hat_2}
    \end{figure}
        \begin{figure}[tb]
      \centering
      \includegraphics[width=\linewidth]{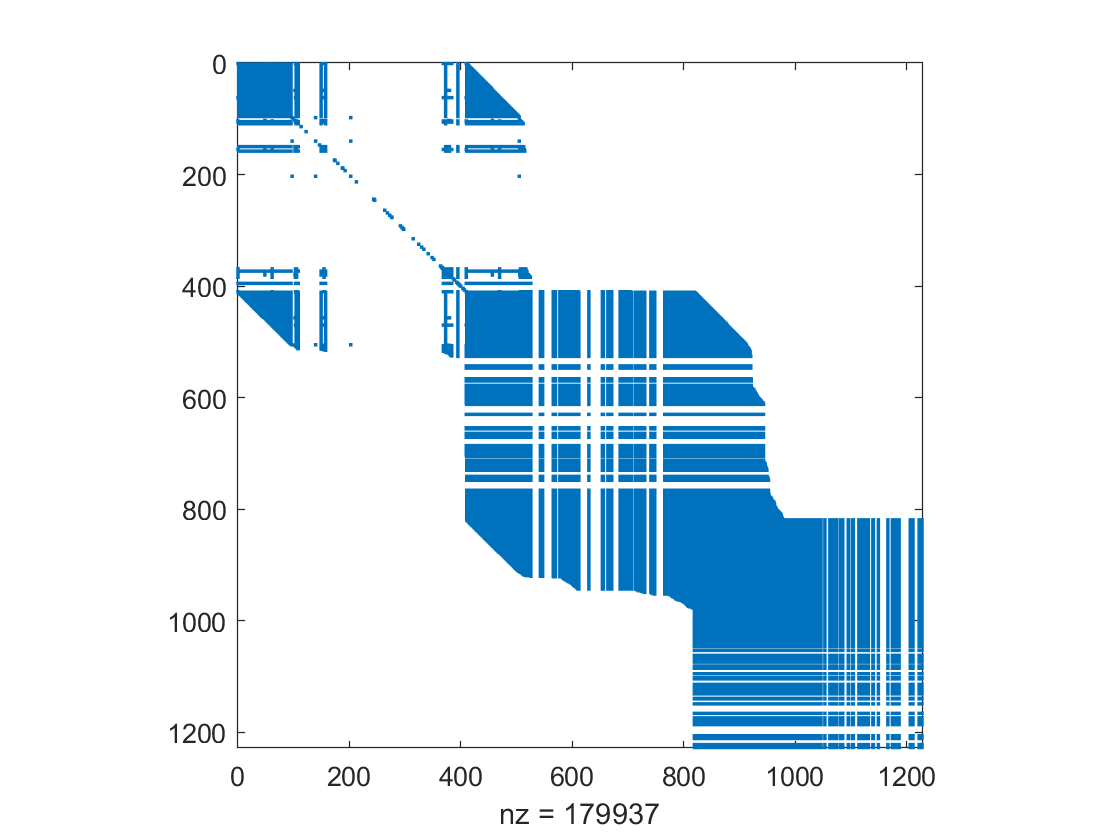}
      \caption{Nonzero Elements of the Reduced Capacitance Matrix in ADC\_Net\_304}\label{Spy_C_hat_2}
    \end{figure}


\section{Algorithm Implementation}
\label{sec:optimization}
This section primarily addresses the densification issue of the reduced-order matrix mentioned at the end of the previous section, and introduces a sparsity control technique and a deflation technique to further improve the multipoint moment matching algorithm.
\subsection{Sparsity Control}
The densification of the first diagonal blocks $G_{p}^{(1)},\vspace{0.5em}C_{p}^{(1)}$ is caused by the elimination process that retains no internal nodes.To address this issue, we consider eliminating internal nodes one by one in the \textbf{AMD}\cite{AMD} order, and stopping the elimination when the number of non-zero elements reaches a certain threshold. In this paper, this threshold is set to 20 times the dimension of the reduced matrix. This process is illustrated in Algorithm 1. Among the input variables, $\eta$ is an artificial variable, which means that the elimination process should stop when the non-zero elements of matrix $\widehat{G}+\widehat{C}$ exceeds $\eta$ times its size. In the following numerical experiment, $\eta$ takes 20.

    \begin{table}[h]
	\centering
	\begin{tabular*}{\hsize}{l}
		\hline
		\textbf{Algorithm 1 Sparse SIP}\\
		\hline
		\textbf{Input} $G$, $C$, $p$, $s$, $\eta$.\\
        1: All nodes are reordered according to the \textbf{AMD} order;\\
        2: Place all ports first;\\
        3: Compute $\widehat{A}= A = G + sC$;\\
        4: $W=I_n$, $\widehat{G}=G$, $\widehat{C}=C$;\\
        5: $k=n$;\\
        6: while $k>p$\\
        7:\qquad if ${\rm nnz}(\widehat{G}+\widehat{C}) > \eta k$\\
        8:\qquad\qquad break;\\
        9:\qquad End if\\
        10:\qquad $\widetilde{A} = \widehat{A}(1:{\rm end}-1,1:{\rm end}-1)$;\\
        11:\qquad $\widetilde{C} = \widehat{C}(1:{\rm end}-1,1:{\rm end}-1)$;\\
        12:\qquad $a=\widehat{A}(1:{\rm end}-1,{\rm end})/\widehat{A}({\rm end},{\rm end})$;\\
        13:\qquad $a_{nn} = \widehat{A}({\rm end},{\rm end})$;\\
        14:\qquad $c=\widehat{C}(1:{\rm end}-1,{\rm end})$;\\
        15:\qquad $c_{nn} = \widehat{C}({\rm end},{\rm end})$;\\
        16:\qquad $\widehat{A} = \widetilde{A}-a_{nn}aa^\top$;\\
        17:\qquad $\widehat{C} = \widetilde{C}-ac^\top -ca^\top +c_{nn}aa^\top$;\\
        18:\qquad $\widehat{G} = \widehat{A} - s\widehat{C}$;\\
        19:\qquad Update $W$: 
        $W = W\begin{bmatrix}
                I_{k-1} & & \\
                -a^\top & 1 &\\
                & & I_{n-k}\\
            \end{bmatrix}$;\\
        \textbf{Output} $\widehat{G}$, $\widehat{C}$, $W$.\\
		\hline
	\end{tabular*}
\end{table} 

    We assume that the number of the remaining nodes in the first elimination is $p_1\geq p$. For the elimination of the subsequent subsystem $\Sigma_{k-1}^{(2)}$:
    \begin{equation}\label{Sigma_k2}
      \Sigma_{k-1}^{(2)}:\left\{\begin{aligned} & G^{(k)} \begin{bmatrix}
                                                      x^{(k)}_{p}(s) \\
                                                      x^{(k)}_{I}(s) 
                                                    \end{bmatrix} + sC^{(k)}\begin{bmatrix}
                                                      x^{(k)}_{p}(s) \\
                                                      x^{(k)}_{I}(s) 
                                                    \end{bmatrix} =\begin{bmatrix}
                                                      -B^{(k)} \\
                                                      0 
                                                    \end{bmatrix}u^{(k)}(s)\\&y^{(k)}(s) = -B^{(k)\top} x^{(k)}_{p}(s)
      \end{aligned}\right.,
    \end{equation}
    where $G^{(k)}$ and $C^{(k)}$ is defined as \cref{G_k+1,C_k+1}, $B^{(k)}\in\mathbb{R}^{p_{k-1}\times p_{k}}$ and $2\leq k\leq m$. The coefficient matrices is obtained via orthogonal transformation and is therefore dense. At this point, the more points are eliminated, the fewer nonzero elements remain. For these systems, we directly retain only the linear ports.

\subsection{Deflation}
From the previous subsection, we know that the last $m-1$ diagonal blocks of the reduced matrices $\widehat{G}$ and $\widehat{C}$ are dense. Therefore, the fewer linear ports there are in the system represented by \cref{Sigma_k2}, the fewer non-zero elements there are in the matrices $G_{p}^{(k)}$ and $C_{p}^{(k)}$ in \cref{Reduced_System}. If internal nodes are eliminated without sparsity control, matrix $C_{C}^{(1)}$ is $n-p$ by $p$. However, we find out that in many applications, matrix $C_{C}^{(1)}$ is nearly rank deficient. To exploit the rank-deficient property of $C_{C}^{(k)}$, the \textbf{RRQR}\cite{RRQR} algorithm can reveal the rank of $C_{C}^{(k)}$, and the cost is only slightly more than the cost of a regular \textbf{QR} factorization. It computes a permutation $P^{(k)}$ and \textbf{QR} factorization:
\begin{equation*}
    C_C^{(k)}P^{(k)}= Q^{(k+1)}\begin{bmatrix}
        R_{11} & R_{12}\\
        &R_{22}
    \end{bmatrix},
\end{equation*}
where \begin{equation*}
    \left\|R_{22}\right\|_2\leq \delta \left\|R_{11}\right\|_2.
\end{equation*}
In this article, $\delta$ takes $1\times 10^{-6}$, then we have the following decomposition:
\begin{equation*}
       C^{(k)}_{C} = Q^{(k+1)}\begin{bmatrix}
                                     B^{(k+1)} \\
                                     0 
                                   \end{bmatrix},\vspace{1em}k=1,2,\cdots,m-1,
     \end{equation*}
     where
     \begin{equation*}
         B^{(k+1)}\equiv\begin{bmatrix}
        R_{11} & R_{12}\\
    \end{bmatrix}.
     \end{equation*}
 Then we have
\begin{equation*}
    p_m\leq p_{m-1} \leq \cdots \leq p_2 < p.
\end{equation*}
It means that the order of the reduced system represented by \cref{Reduced_System} is less than $mp$. Compared with the system generated by \textbf{TurboMOR}, whose order is exactly $mp$, the impact of an increase in the number of ports on algorithm efficiency will be significantly reduced.

Taking the case ADC\_Net\_27 as an example, this test case has 304 ports. Here, only 2 frequency points are selected:
\begin{equation*}
    S=[0;1\times 10^6];
\end{equation*}
If the deflation technique is not introduced, the reduced-order matrices obtained through multi-frequency-point elimination is shown in \cref{Spy_G_hat_W,Spy_C_hat_W}. When we apply deflation, the reduced-order matrices obtained through multi-frequency-point elimination is shown in \cref{Spy_G_hat_D,Spy_C_hat_D}.
\begin{figure}[tb]
      \centering
      \includegraphics[width=\linewidth]{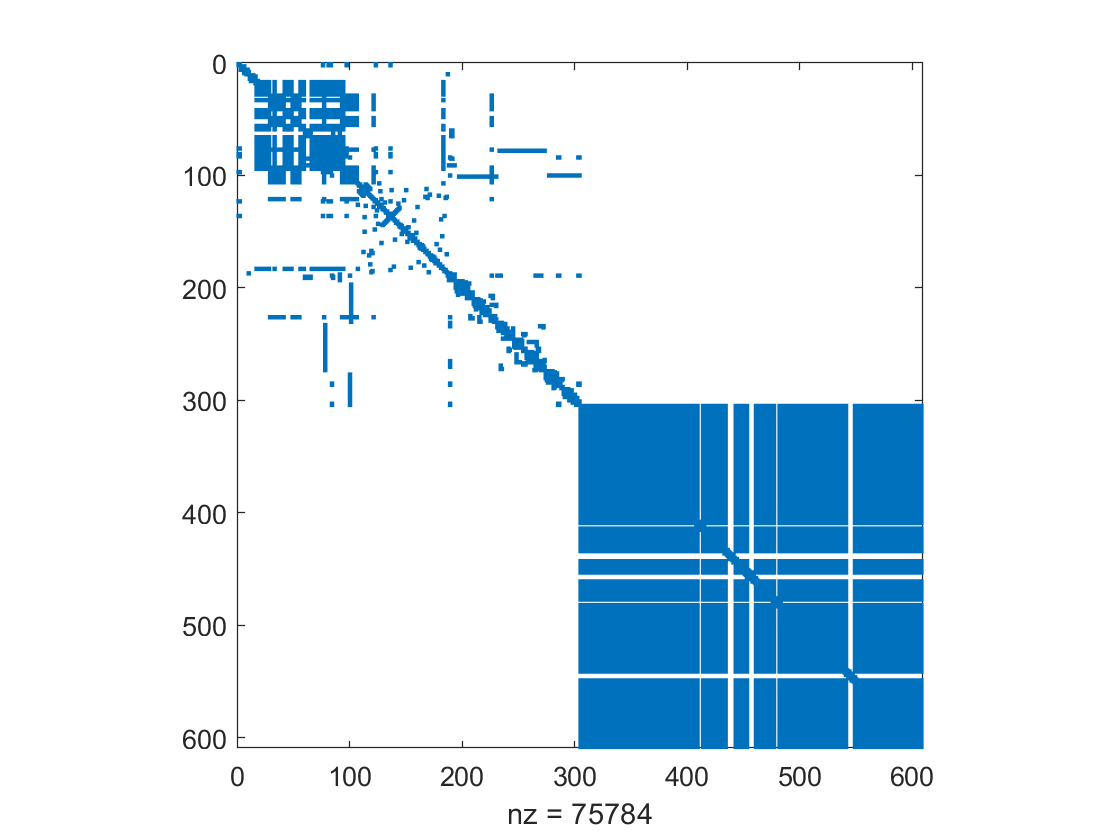}
      \caption{Nonzero of the Conductance Matrix without Deflation in ADC\_Net\_27}\label{Spy_G_hat_W}
    \end{figure}
    \begin{figure}[tb]
      \centering
      \includegraphics[width=\linewidth]{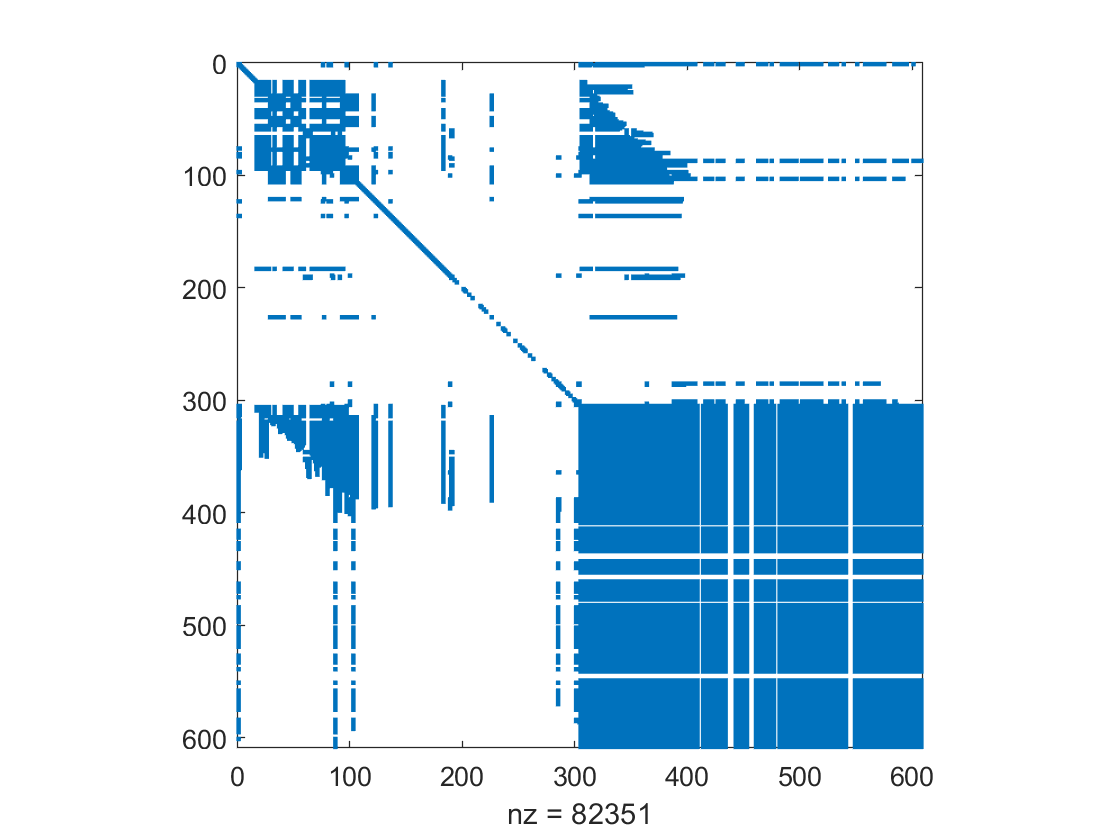}
      \caption{Nonzero of the Capacitance Matrix without Deflation in ADC\_Net\_27}\label{Spy_C_hat_W}
    \end{figure}
\begin{figure}[tb]
      \centering
      \includegraphics[width=\linewidth]{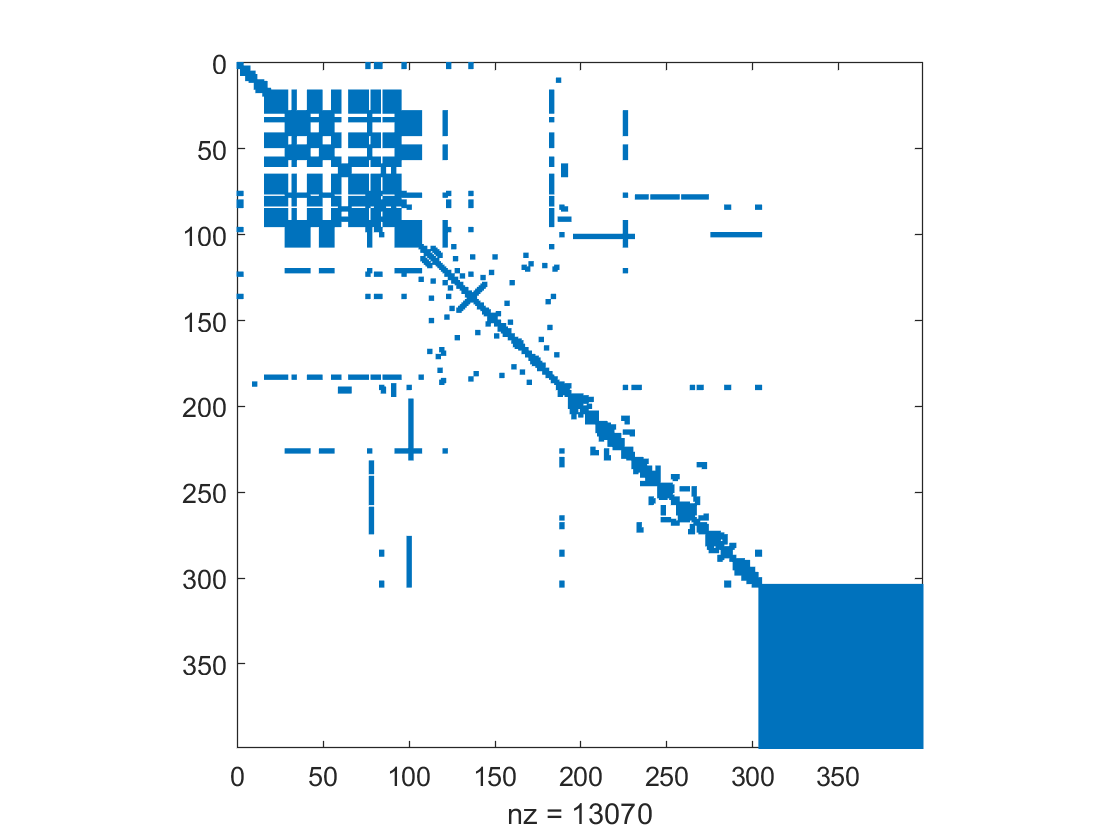}
      \caption{Nonzero of Conductance Matrix with Deflation in ADC\_Net\_27}\label{Spy_G_hat_D}
    \end{figure}
    \begin{figure}[tb]
      \centering
      \includegraphics[width=\linewidth]{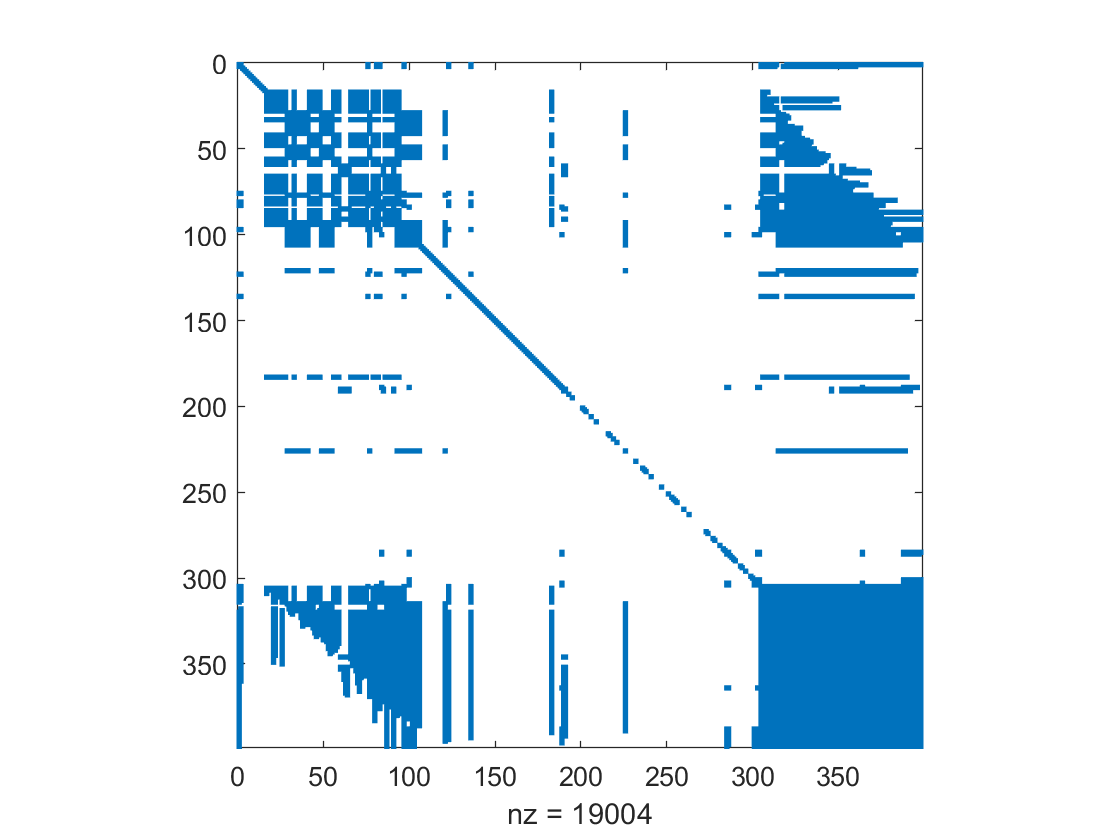}
      \caption{Nonzero of Capacitance Matrix with Deflation in ADC\_Net\_27}\label{Spy_C_hat_D}
    \end{figure}
By comparing the reduced-order matrices obtained under the two scenarios, we find that the matrix block $\widehat{C}_{I}^{(2)}$ is dense. After deflation is performed, the size of this block is significantly reduced, which in turn leads to a substantial decrease in the number of non-zero elements of the entire matrix.
\subsection{Implementation}
Based on multipoint reduction, we can obtain the frame by combining sparsity control technique and deflation in Algorithm 2.
     \begin{table}[h]
	\centering
	\begin{tabular*}{\hsize}{l}
		\hline
		\textbf{Algorithm 2 Sparse Multi-Point Matching Method}\\
		\hline
		\textbf{Input} $G$, $C$, $B$, $S$ formulated by \cref{GC,B,S}.\\
        1: $[\widehat{G}^{(1)}_{p},\widehat{C}^{(1)}_{p},W_1]$ = \textbf{Sparse SIP}$(G,C,p,s_1)$;\\
        2: Compute $\widehat{G}^{(1)}=W_1^\top G W_1, \widehat{C}^{(1)} = W_1^\top C W_1$\\
        3: $p_1$ is the size of $\widehat{G}^{(1)}_{p}$;\\
        4: \textbf{if} $m\geq 2$\\
        5:\quad \textbf{for} $i=2:m$\\
        6:\qquad Obtain $\widehat{C}^{(i-1)}_{C}$, $\widehat{C}^{(i-1)}_{I}$ and $\widehat{G}^{(i-1)}_{I}$ formulated in \cref{Gk,Ck};\\
        7:\qquad Apply \textbf{RRQR} to compute $C^{(i-1)}_{C} = Q_i\begin{bmatrix}
                                             B^{(i)} \\
                                             0 
                                           \end{bmatrix};$\\
        8:\qquad Obtain $G^{(i)}$ and $C^{(i)}$ by \cref{G_k+1,C_k+1};\\
        9:\qquad $p_i$ is the columns of $B^{(i)}$;\\ 
        10:\qquad $[\widehat{G}^{(i)}_{p},\widehat{C}^{(i)}_{p},W_i]$ = \textbf{SIP}$(G^{(i)},C^{(i)},p_i,s_i)$;\\
        11:\qquad Obtain $G^{(i)}$ and $C^{(i)}$:\\
        12:\qquad $\widehat{G}^{(i)} = W_i^\top G^{(i)} W_i$, $\widehat{C}^{(i)} = W_i^\top C^{(i)} W_i$;\\
        13:\quad \textbf{End for}\\
        14:\quad Construct $\widehat{G}$ and $\widehat{C}$ as \cref{Reduced_System};\\
        15:\quad $k=\sum_{i=1}^{m}p_i$;\\
        16:\quad $\widehat{B}=B(1:k,:)$;\\
        17: \textbf{Else}\\
        18:\quad $\widehat{G} = \widehat{G}^{(1)}_{p}$, $\widehat{C}= \widehat{C}^{(1)}_{p}$, $\widehat{B}=B(1:p_1,:)$;\\
        19: \textbf{End if}\\
        \textbf{Output} $\widehat{G}$, $\widehat{C}$, $\widehat{B}$.\\
		\hline
	\end{tabular*}
\end{table} 

In the algorithm, the implementation process of the \textbf{RRQR} algorithm is detailed in the algorithm  of section 3 in the article \cite{RRQR} and  that of the \textbf{SIP} algorithm is detailed in the algorithm 2 in section 4 in the article \cite{SIP}. As shown in \cref{Spy_G_hat_2_opt}, in the reduced-order system obtained by the optimized algorithm, $G^{(1)}_{p}$ and $C^{(1)}_{p}$ are larger sparse matrices, and other diagonal blocks are smaller dense matrix. Compare \cref{Spy_G_hat_2_opt} with \cref{Spy_G_hat_2}, the number of nonzero elements decreased by more than 70\%.
    \begin{figure}[tb]
      \centering
      \includegraphics[width=\linewidth]{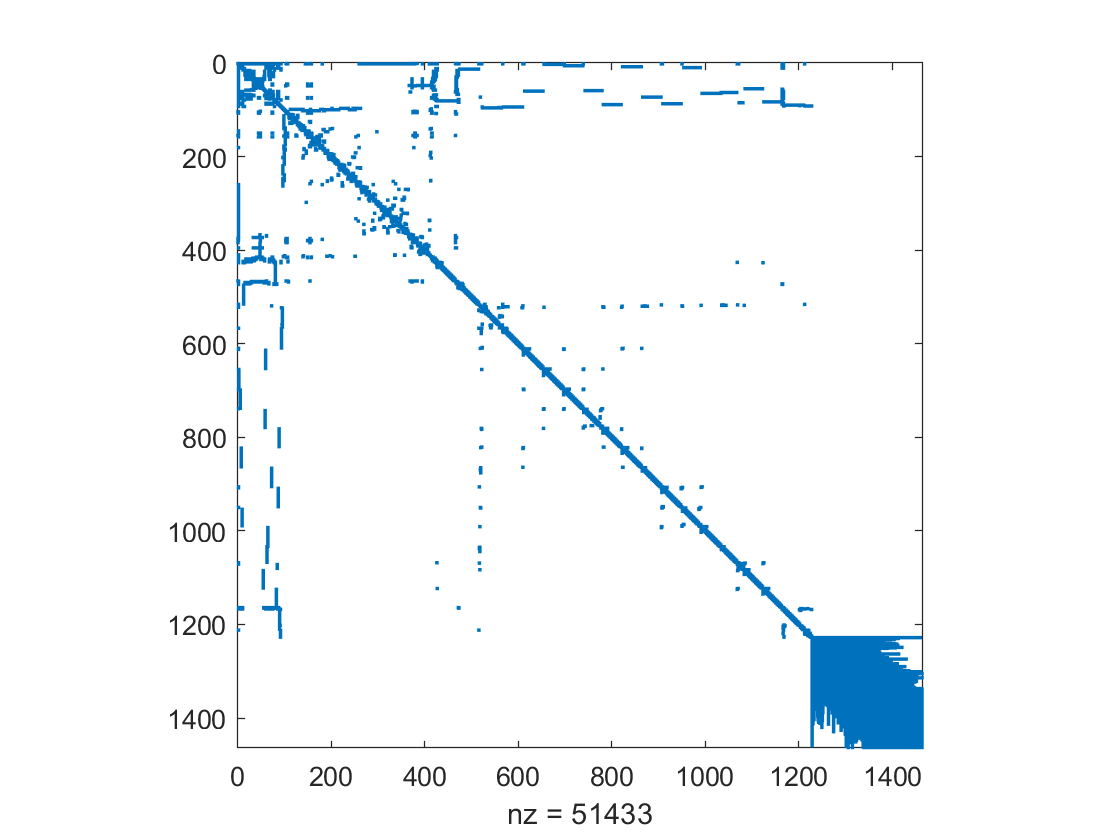}
      \caption{Nonzero Elements of the Optimized Reduced Conductance Matrix in ADC\_Net\_304}\label{Spy_G_hat_2_opt}
    \end{figure}
    \begin{figure}[tb]
      \centering
      \includegraphics[width=\linewidth]{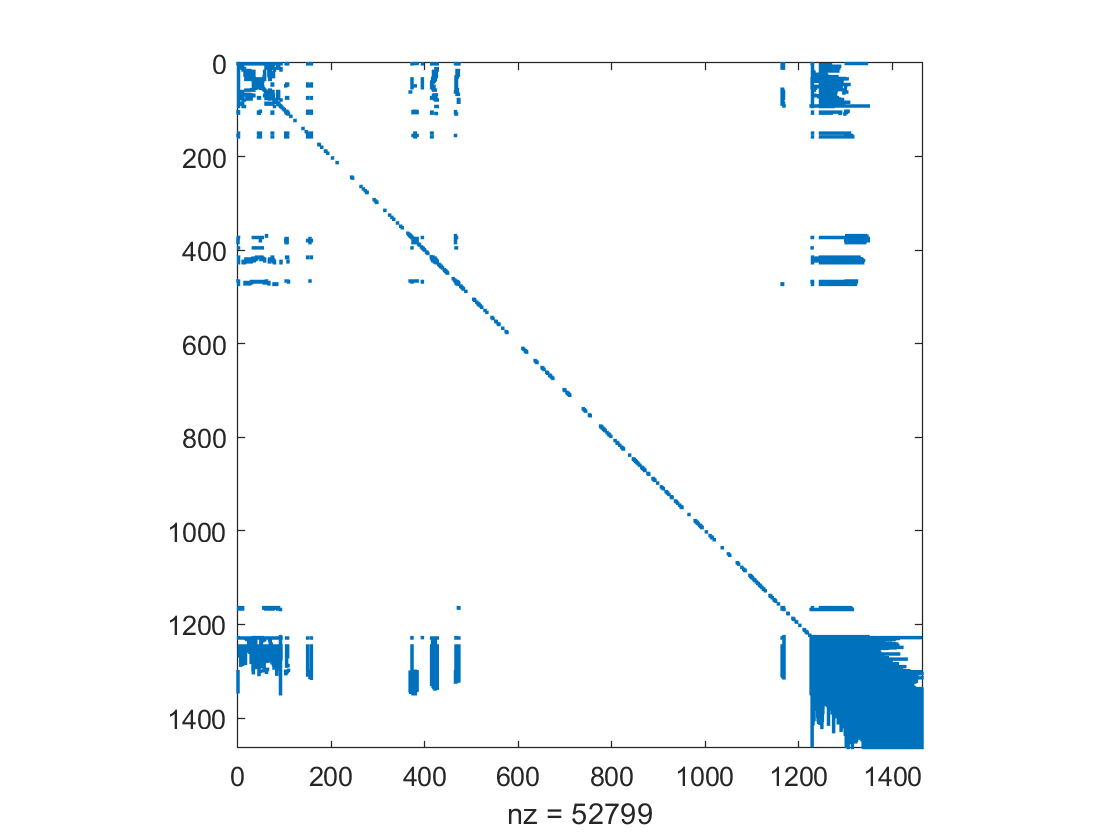}
      \caption{Nonzero Elements of the Optimized Reduced Capacitance Matrix in ADC\_Net\_304}\label{Spy_C_hat_2_opt}
    \end{figure}



\section{Numerical Results}
\label{sec:results}
In this section, we will verify the accuracy of our proposed method and the efficiency of deflation. Afterwards, we will compare the reduction efficiency with \textbf{SIP}, \textbf{PRIMA} and \textbf{TurboMOR}. Test cases are provided by our industrial partner, as shown in Table \ref{base}. The \#NNZ item is obtained from the explicit system matrix $G+C$. For each resolution circuit, VDD and VSS networks are in the same SPICE file. All experiments are carried out using a single CPU core of a computer cluster running 64-bit CentOS 7.0 with Intel Xeon Gold 6226R @2.90 GHz and 700GB Memory. The wall-clock runtime is reported. 

To reduce the fill-ins introduced by \textbf{SIP} method, we employ \textbf{AMD}\cite{AMD} ordering for the nodes and control the number of eliminated nodes to ensure the sparsity of the reduced system.
\begin{table}[tb]
	\centering
	\caption{INFORMATION OF Original Systems.}
    \begin{tabular}{c c c c c }
    \hline
    Case & Nodes &  Ports &  nnz($G+C$) \\
    \hline
    ADC\_Net\_27 & 3029 & 304 & 19663 \\
    \hline
    ADC\_Net\_304 & 4374 & 409 & 28808 \\
    \hline
    PLL\_Net\_1 & 9500 & 924 & 38790  \\
    \hline
    ADC\_Net\_64 & 11070 & 1082 & 71200 \\
    \hline
    PLL\_Net\_9 & 23132 & 360 & 106532 \\
    \hline 
    PLL\_Net\_301 & 1308 & 59 & 5848 \\
    \hline
    ADC\_Net\_1 & 2044 & 75 & 13400 \\
    \hline
	\end{tabular}
	\label{base}
\end{table} 

\subsection{Verify Accuracy} 
In this subsection, we mainly compare the accuracy with \textbf{SIP}, and \textbf{TurboMOR-RC}. We define the relative errors on frequency $f$:
    \begin{equation}\label{E}
      E_R(f)\equiv\frac{\left\|H(f)-\widehat{H}(f)\right\|_2}{\left\|H(f)\right\|_2},
    \end{equation}
    \begin{equation}\label{E}
      E_C(f)\equiv\frac{\left\|H(2\pi \jmath f)-\widehat{H}(2\pi\jmath f)\right\|_2}{\left\|H(2\pi\jmath f)\right\|_2},
   \end{equation}
    where $f\in[0,10^{13}]$. $E_R$ can verify the multi-point moment matching property of our proposed method.  $E_C$ can simulate the relative error of node voltages in simulation to a certain extent. 

    Firstly, we compare the accuracy of the reduced systems generated by \textbf{SIP} and our proposed method. The frequency points are 
    \begin{equation*}
        s = [0,1\times 10^9,1\times 10^{12}].
    \end{equation*}
    The relative errors $E_R$ at different frequency points are illustrated in \cref{Error_SP_1,Error_SP_2}. The relative errors $E_C$ at different frequency points are illustrated in \cref{Error_SP_3,Error_SP_4}. As can be seen from the figures, for the two cases with a few number of ports, \textbf{SIP} is insufficient to meet the accuracy requirements in the high-frequency band.
    \begin{figure}[tb]
      \centering
      \includegraphics[width=\linewidth]{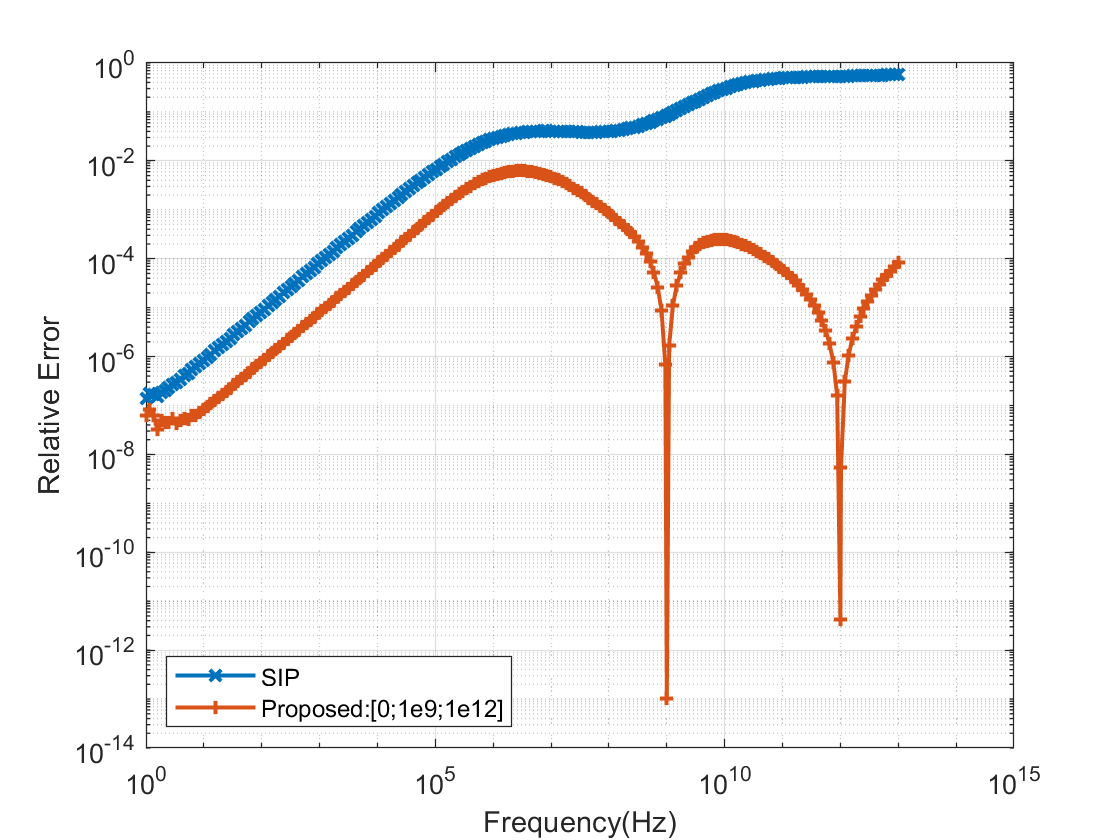}
      \caption{$E_R$ Comparison with \textbf{SIP} at different frequency in PLL\_Net\_301}
      \label{Error_SP_1}
    \end{figure}
     \begin{figure}[tb]
      \centering
      \includegraphics[width=\linewidth]{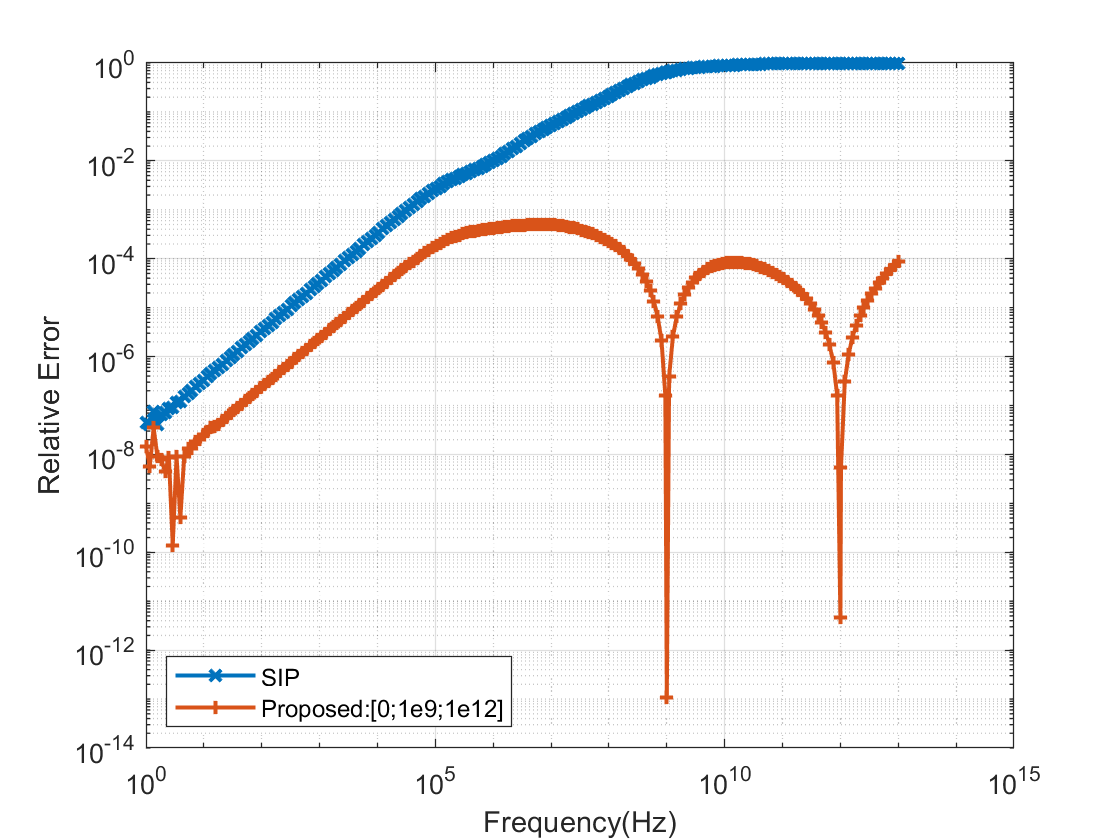}
      \caption{$E_R$ Comparison with \textbf{SIP} at different frequency in ADC\_Net\_1}
      \label{Error_SP_2}
    \end{figure}
    \begin{figure}[tb]
      \centering
      \includegraphics[width=\linewidth]{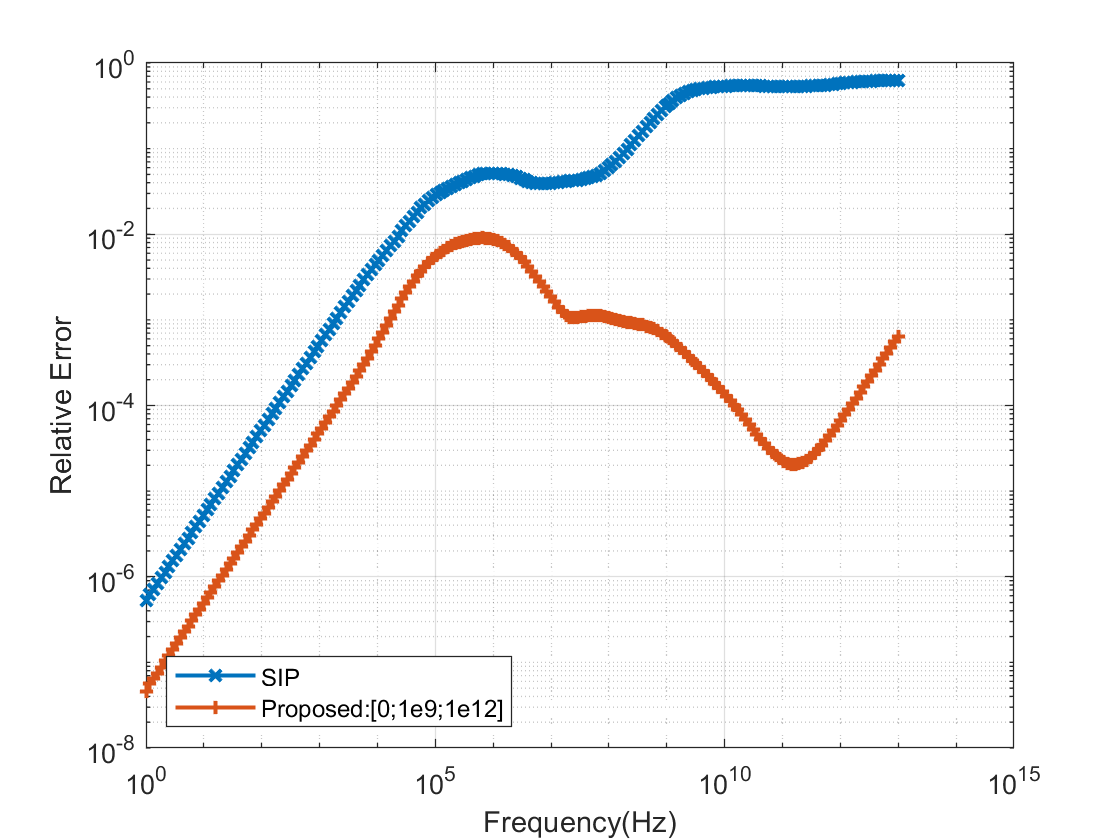}
      \caption{$E_C$ Comparison with \textbf{SIP} at different frequency in PLL\_Net\_301}
      \label{Error_SP_3}
    \end{figure}
     \begin{figure}[tb]
      \centering
      \includegraphics[width=\linewidth]{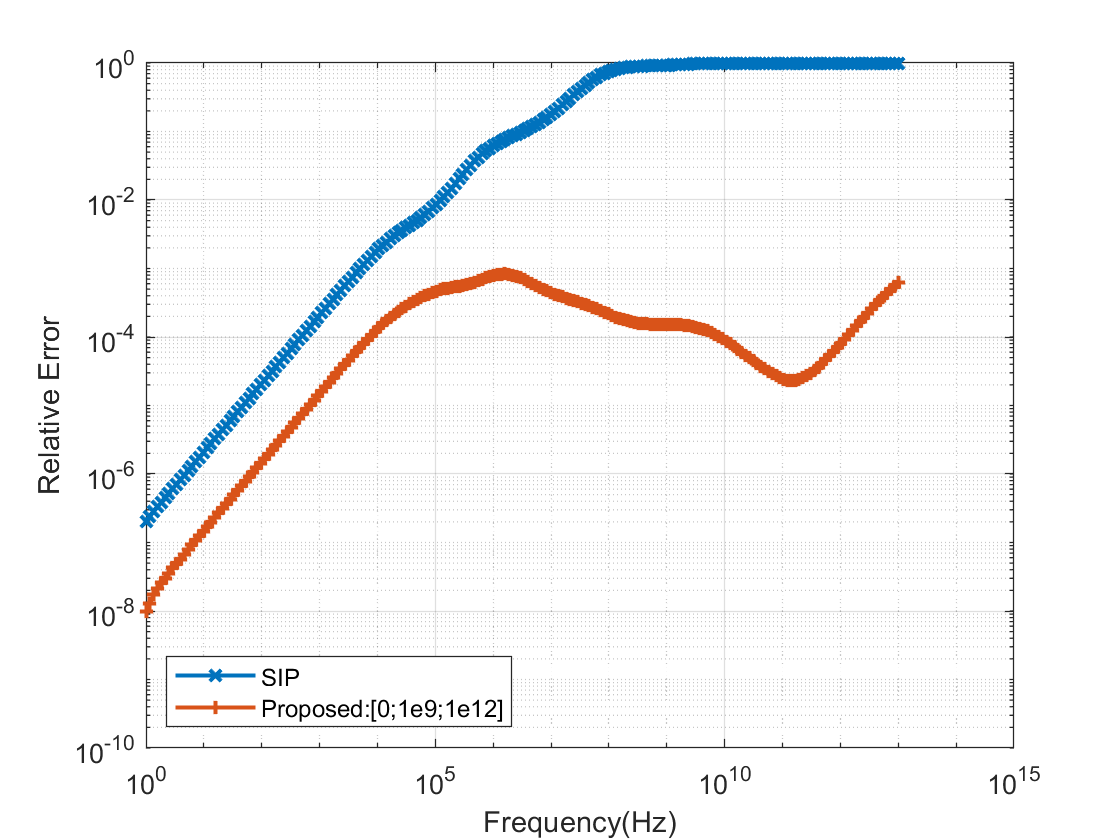}
      \caption{$E_C$ Comparison with \textbf{SIP} at different frequency in ADC\_Net\_1}
      \label{Error_SP_4}
    \end{figure}
    
    When the order of moment matching in \textbf{TurboMOR-RC} takes 8 and the frequency points are 
    \begin{equation*}
        s = [0,0,1\times 10^9,1\times 10^{12}].
    \end{equation*} The relative error at different frequency points are illustrated in \cref{Error_TP_1,Error_TP_2}. The relative errors $E_C$ at different frequency points are illustrated in \cref{Error_TP_3,Error_TP_4}. It can be seen from this that both methods achieve high accuracy in the lower frequency range. However, when the frequency reaches a certain level, our method exhibits higher accuracy compared to \textbf{TurboMOR}.
    \begin{figure}[tb]
      \centering
      \includegraphics[width=\linewidth]{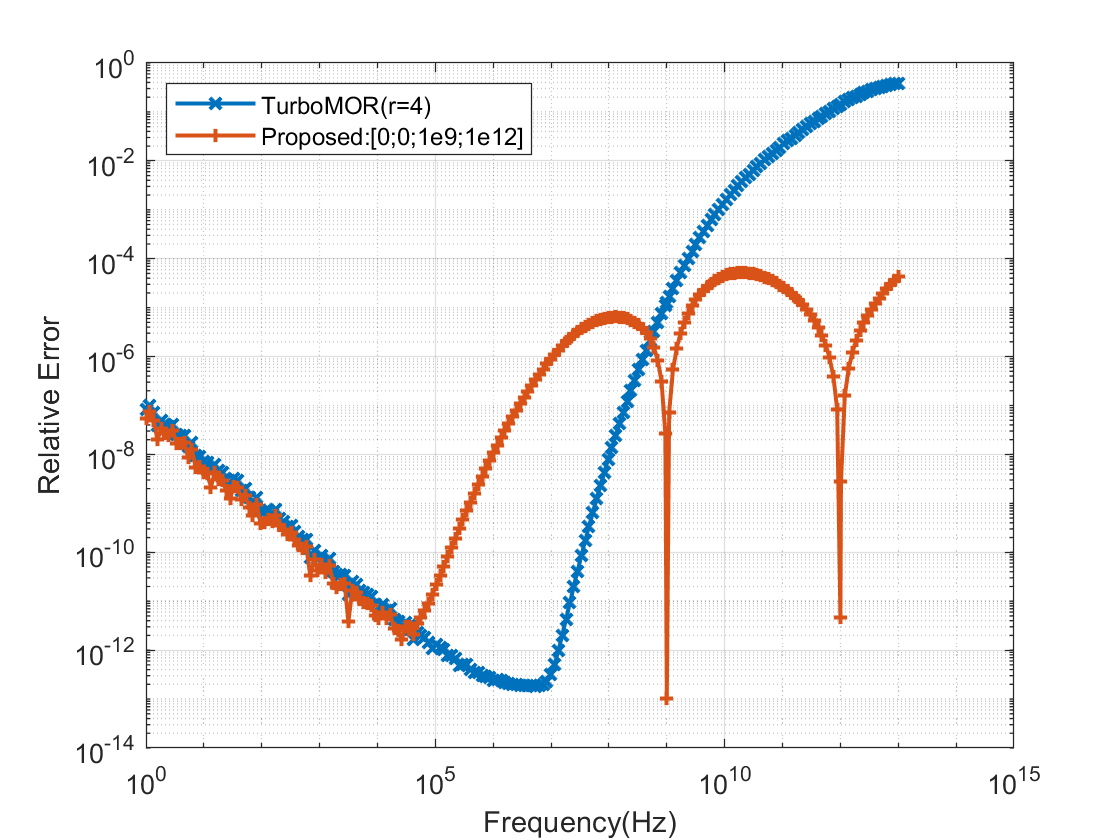}
      \caption{$E_R$ Comparison with \textbf{TruboMOR} at different frequency points in PLL\_Net\_301(r=4)}
      \label{Error_TP_1}
    \end{figure}
    \begin{figure}[tb]
      \centering
      \includegraphics[width=\linewidth]{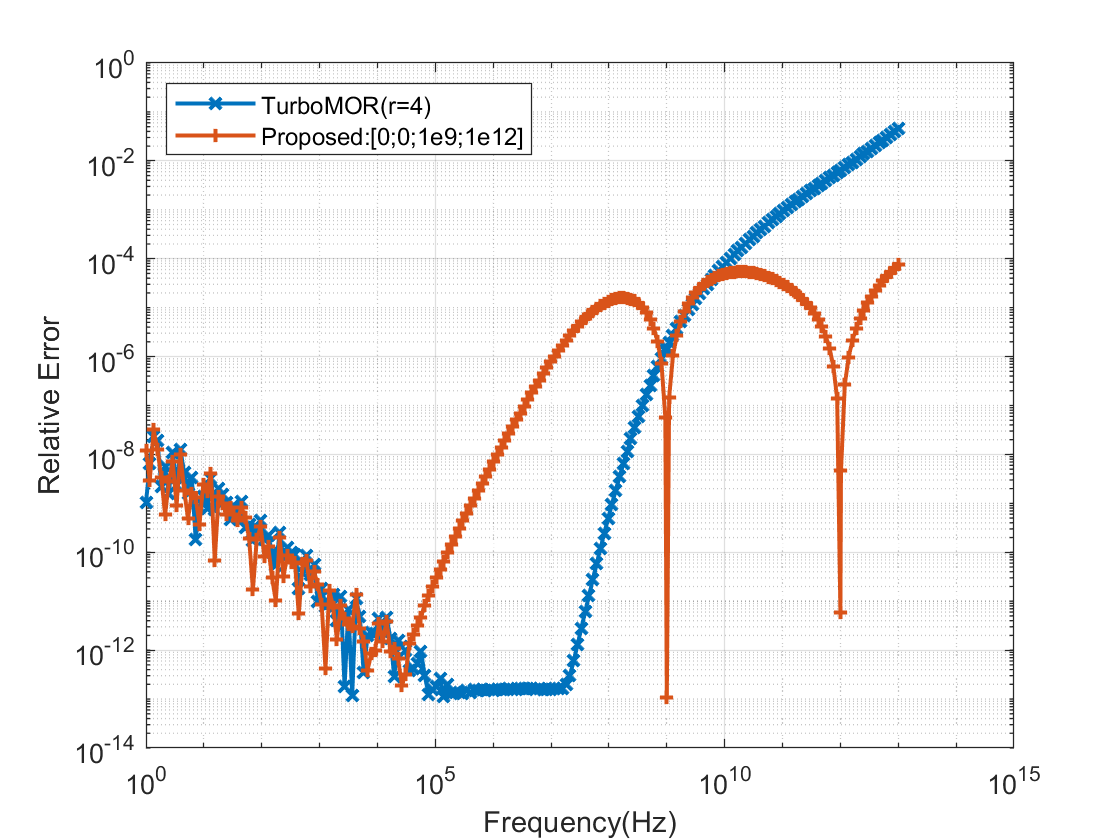}
      \caption{$E_R$ Comparison with \textbf{TruboMOR} at different frequency points in ADC\_Net\_1(r=4)}
      \label{Error_TP_2}
    \end{figure}
    \begin{figure}[tb]
      \centering
      \includegraphics[width=\linewidth]{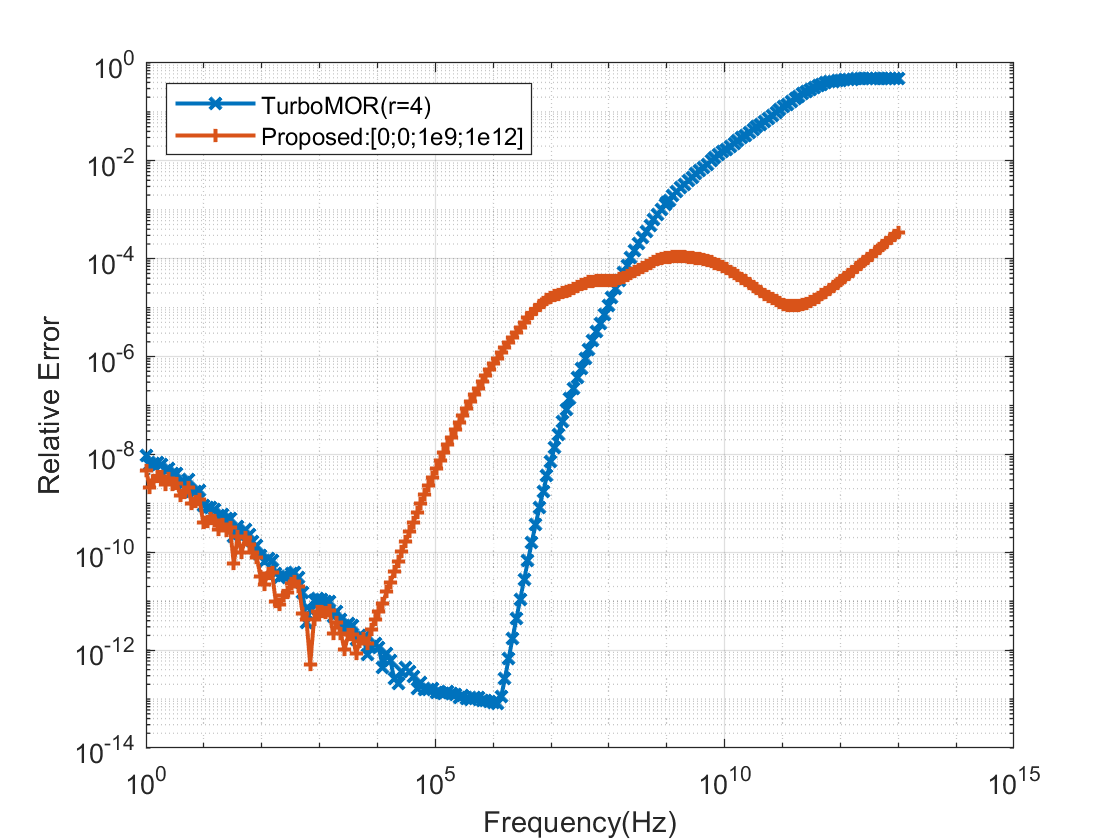}
      \caption{$E_C$ Comparison with \textbf{TruboMOR} at different frequency points in PLL\_Net\_301(r=4)}
      \label{Error_TP_3}
    \end{figure}
    \begin{figure}[tb]
      \centering
      \includegraphics[width=\linewidth]{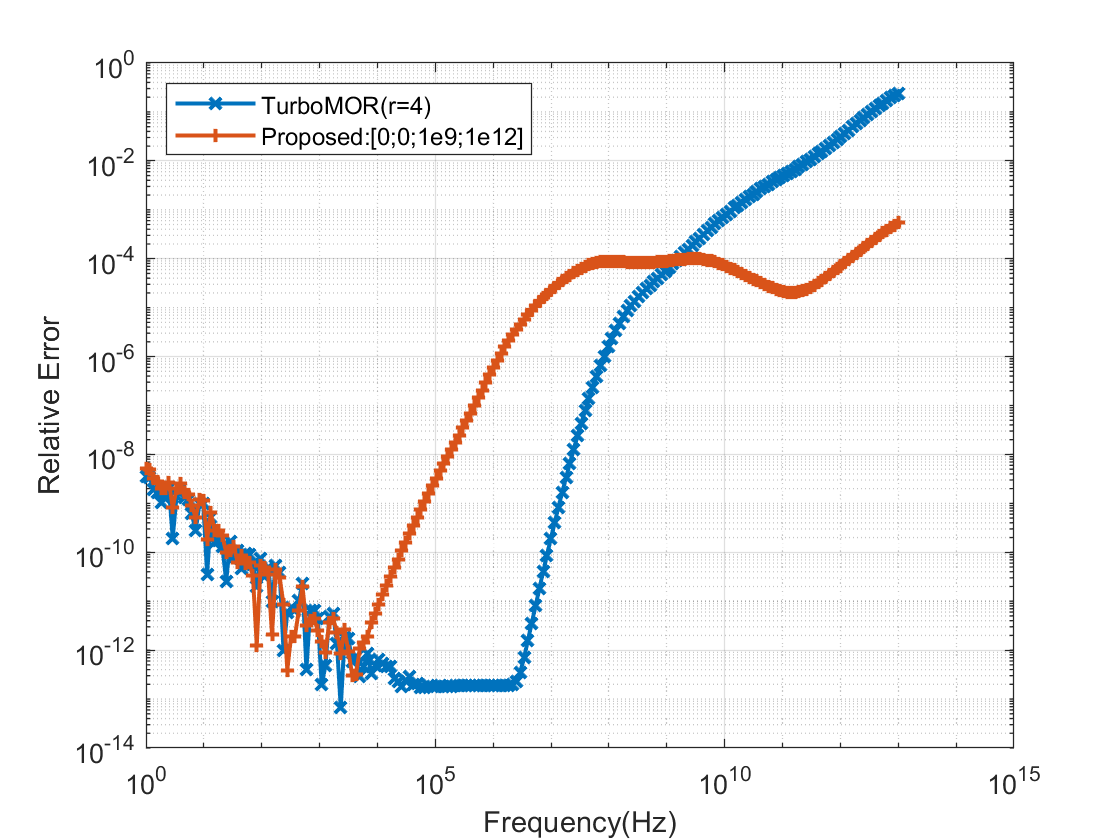}
      \caption{$E_C$ Comparison with \textbf{TruboMOR} at different frequency points in ADC\_Net\_1(r=4)}
      \label{Error_TP_4}
    \end{figure}

    In addition, \cref{Error_SP_1,Error_SP_2,Error_TP_1,Error_TP_2} reveals that the moment matching properties at $1\times 10^9$ and $1\times 10^{12}$ are satisfied.

\subsection{Verify Deflation} 
To verify the significance of the deflation technique in multi-port cases (the first 5 cases in \cref{base}), this subsection compares the number of non-zero elements and accuracy of the reduced-order systems obtained with and without deflation during decomposing matrices $C_{C}^{(k)}$. The frequency points are 
    \begin{equation*}
        s = [0,1\times 10^6].
    \end{equation*} The results are shown in \cref{Deflation}. The "Reduction Ratio" column indicates the percentage reduction in the number of non-zero elements deflation is  applied. As can be seen from the table, for these test cases, applying deflation significantly reduces the number of non-zero elements in the reduced-order matrix, thereby improving the efficiency of model order reduction. Meanwhile, there is no substantial loss in the accuracy of the reduced-order system.

    However, for a specific case (PLL\_Net\_9), if sparsity control is not implemented, the number of non-zero elements remains unacceptable. In such case we need to apply sparsity control to ensure that our reduced systems are sparse. Then the results of \textbf{RC} reduction for PLL\_Net\_9 are shown in\cref{Deflation_5}. We can observe that when both techniques are used simultaneously, the number of nonzero elements is significantly reduced. In  next subsection, we will compare our implemented method with our \textbf{RCR} methods.
\begin{table}[h]
	\centering
	\caption{Efficiency of Implemntation technique in PLL\_Net\_9.}
    \begin{tabular}{c m{1.5cm}<{\centering} m{1cm}<{\centering}  m{1cm}<{\centering} m{1.5cm}<{\centering} m{1cm}<{\centering} }
    \hline
    Ports &  Whether Deflation & Whether Sparsity Control & Nodes &  nnz($\widehat{G}+\widehat{C}$) & Reduction Ratio \\
    \hline
      \multirow{3}{*}{360} & No & No & 720 & 386,292 & ———— \\ 
                       &   Yes  & No  & 713 & 381,271 & 1.30\%  \\
                       &   Yes  & Yes  & 12830 & 135830 & 64.84\%  \\
    \hline
	\end{tabular}
	\label{Deflation_5}
\end{table}  
\begin{table*}[h]
	\centering
	\caption{Deflation Efficiency of Different Cases.}
    \begin{tabular}{c m{1cm}<{\centering} m{1.5cm}<{\centering}  m{1cm}<{\centering} m{1.5cm}<{\centering}  m{1.5cm}<{\centering} m{1.5cm}<{\centering} m{1.7cm}<{\centering} }
    \hline
    Case Name & Ports &  Whether Deflation & Num. Nodes &  nnz($\widehat{G}+\widehat{C}$) & Reduction Ratio  & $E_C(1e6)$ & $E_C(1e12)$\\
    \hline  
    \multirow{2}{*}{ADC\_Net\_27} & \multirow{2}{*}{304} & No & 608 & 83,522 & ———— & 1.2053e-9 & 4.3066e-8 \\
                       &                     & Yes  & 398  & 19,914 & 76.16\%  & 1.2161e-9 & 4.4755e-8 \\
    \hline
    \multirow{2}{*}{ADC\_Net\_304} & \multirow{2}{*}{409} & No & 818 & 102,660 &  ———— & 1.3710e-10 & 1.7824e-9 \\
                       &                     & Yes  & 526 & 37,058  & 63.90\%  & 1.4308e-10 & 1.8241e-9 \\
    \hline
    \multirow{2}{*}{PLL\_Net\_1} & \multirow{2}{*}{1140} & No & 1,848 & 892,946  & ————  & 2.7951e-9 & 1.0615e-8 \\
                       &                     & Yes  & 1,140 & 242,540 & 72.84\%  & 2.8844e-9 & 5.9910e-8 \\
    \hline
    \multirow{2}{*}{ADC\_Net\_64} & \multirow{2}{*}{1082} & No & 2,164 & 1,035,476  & ————  & 5.6848e-9 & 1.0615e-5 \\
                       &                     & Yes  & 1,486 & 281,900 & 72.87\%  & 5.5547e-9 & 1.1015e-5 \\
    \hline
    \multirow{2}{*}{PLL\_Net\_9} & \multirow{2}{*}{360} & No & 720 & 386,292 & ————  & 3.4303e-8 & 7.0083e-5 \\
                       &                     & Yes  & 713 & 381,271 & 1.30\%  & 3.4303e-8 & 7.0238e-5\\
    \hline
	\end{tabular}
	\label{Deflation}
\end{table*} 

\subsection{Compare Efficiency} 
In this subsection, we mainly test the model order reduction efficiency of different algorithms for multi-port RC circuits. We need to compare both the accuracy at different frequencies and the simulation time of the reduced systems obtained by using the three algorithms. In this case, the simulation time refers to the time required to compute the transfer function. Computing the transfer function involves performing calculations such as matrix multiplications and solving linear equations. When computing $\left(G+s C\right)^{-1}B,$ We utilize the \textbf{KLU}\cite{KLU} library to perform \textbf{LU} factorization and then employ forward substitution and back substitution to solve the $p$ linear systems. The order of moment matching in \textbf{PRIMA} and \textbf{TurboMOR-RC} takes 4 and the frequency points are 
    \begin{equation*}
        s = [0,1\times 10^6].
    \end{equation*}
The results are revealed in \cref{efficiency}. The simulation time refers to the time taken to compute the transfer functions at 100 frequency points. The "Speed up" column represents the improvement in simulation speed of the reduced-order systems obtained by other algorithms relative to those obtained by the \textbf{PRIMA} algorithm.

As can be seen from the table, for multi-port \textbf{RC} circuits, the reduced-order systems obtained by \textbf{TurboMOR} and \textbf{PRIMA} contain an excessive number of non-zero elements, which leads to increased simulation time. In contrast, our method not only achieves model order reduction but also maintains an accuracy level comparable to that of these two methods. 

Compared with \textbf{SIP}, without introducing a significant number of additional linear components (which would undermine the efficiency gains from order reduction), our method can better retain the high-frequency information of the original system, thereby achieving a substantial improvement in accuracy at high-frequency points.

\begin{table*}[tb]
	\centering
	\caption{INFORMATION OF REDUCED SYSTEMS OBTAINED BY DIFFERENT METHODS.}
    \begin{tabular}{c m{1cm}<{\centering} m{1.5cm}<{\centering}  m{1cm}<{\centering} m{1.5cm}<{\centering}  m{1.5cm}<{\centering} m{1.5cm}<{\centering} m{1.7cm}<{\centering} m{1.7cm}<{\centering}}
    \hline
    Case Name & Original Time(s) &  Method & Num. Nodes &  nnz($G+C$) & Simulation Time(s) &Speed up & $E_C(1e9)$ & $E_C(1e12)$\\
    \hline  
    \multirow{4}{*}{ADC\_Net\_27} & \multirow{4}{*}{0.96} & PRIMA & 606 & 369,664 & 15.04  & 1$\times$ & 6.3219e-12 & 1.6467e-8 \\
                       &                     & TurboMOR  & 608 & 31,036 & 0.84  & 17.90$\times$ & 6.3099e-12 & 1.6316e-8 \\
                       &                     & SIP & 304 & 4,234 & 0.20 & 75.20$\times$ & 2.4359e-6 & 8.6037e-3\\
                       &                     & Proposed & 398 & 10,364 & 0.32 & 47.00$\times$ & 6.3120e-12 & 1.3151e-8\\
    \hline    
    \multirow{4}{*}{ADC\_Net\_304} & \multirow{4}{*}{1.80} & PRIMA & 806 & 649,636 & 32.28  & 1$\times$ & 3.2083e-11 & 1.1742e-9 \\
                       &                     & TurboMOR  & 818 & 129,852 & 6.51  & 4.97$\times$ & 3.9964e-11 & 1.3767e-9 \\
                       &                     & SIP & 2,045 & 14,019  & 0.94 &50.59$\times$ & 1.9480e-6 & 2.0618e-3\\
                       &                     & Proposed & 2,143 & 27,255 & 1.28 & 25.29$\times$ & 2.6984e-11 & 3.1362e-10\\
    \hline 
    \multirow{4}{*}{PLL\_Net\_1} & \multirow{4}{*}{9.22} & PRIMA & 1,358 & 1844,164  & 199.19  & 1$\times$ & 3.8136e-10 & 5.6474e-8 \\
                       &                     & TurboMOR  & 1,848 & 309,387 & 22.51  & 8.84$\times$ & 3.9934e-10 & 4.7658e-8 \\
                       &                     & SIP & 3,696 & 55,584 & 5.80 & 34.34$\times$ & 6.4044e-6 & 1.1701e-3\\
                       &                     & Proposed & 3,915 & 94,395 & 7.08 & 28.13$\times$ & 3.3859e-10 & 3.3253e-11\\
    \hline
    \multirow{4}{*}{ADC\_Net\_64} & \multirow{4}{*}{12.15} & PRIMA & 2,158 & 4,656,964 & 710.23  & 1$\times$ & 2.0020e-11 & 1.0234e-5  \\
                       &                     & TurboMOR  & 2,164 & 161,931 & 14.14  & 50.22$\times$ & 2.0474e-11 & 1.0234-5 \\
                       &                     & SIP & 3,246 & 27,936  & 4.11 & 172.80$\times$ & 2.1896e-5 & 1.3537e-2\\
                       &                     & Proposed & 3601 & 48,639  & 5.73 & 123.95$\times$ & 2.0005e-11 & 1.0229e-5\\
    \hline  
    \multirow{4}{*}{PLL\_Net\_9} & \multirow{4}{*}{9.95} & PRIMA & 714 & 509,796 & 33.06  & 1$\times$ & 4.0187e-10 & 4.2888e-5  \\
                       &                     & TurboMOR  & 720 & 384,770 & 14.14  & 50.22$\times$ & 1.0103e-10 & 4.2888-5 \\
                       &                     & SIP & 12,600 & 77,084  & 5.61 & 172.80$\times$ & 6.2099e-5 & 7.7748e-3\\
                       &                     & Proposed & 12,830 & 135,830  & 7.74 & 123.95$\times$ & 5.4306e-11 & 3.7866e-5\\
    \hline  
	\end{tabular}

	\label{efficiency}
\end{table*} 

\section{Conclusion}
\label{sec:conclu}
To conclude, we propose a multipoint moment matching framework and introduce sparsity control and deflation techniques to guarantee sparsity. Compared to \textbf{SIP}, our method significantly improves the accuracy of the reduced-order system at high-frequency points without introducing a large number of additional linear components. Compared to other high-order moment matching methods, our approach significantly reduces the number of non-zero elements in the reduced system, thereby greatly improving simulation speed. This makes our method particularly valuable in applications where high-frequency precision is a key requirement.


\appendices
\section{Details of the Multi-Frequency Points Reduction Process}
\label{app:process}
In this section, we will present the detailed iterative process of multi-frequency-point reduction, where the number of iterations $k$ is greater than 1. In the $k$-th iteration, define matrix $A^{(k)}$:
     \begin{equation*}
      A^{(k)}\equiv G^{(k)}+s_k C^{(k)}=\begin{bmatrix}
          A^{(k)}_{p} & A^{(k)\top}_{C} \\
          A^{(k)}_{C} & A^{(k)}_{I} 
        \end{bmatrix},\vspace{1em} k\geq 2.
    \end{equation*}
    We can construct the congruence transformation $W_k$ to decouple the linear ports and internal nodes in system $\Sigma_{k-1}^{(2)}$:
    \begin{equation*}
      W_k =  \begin{bmatrix}
                 I_{p_k} &  \\
                 -\left(A_{I}^{(k)}\right)^{-1}A^{(k)}_{C} & I 
               \end{bmatrix};
    \end{equation*}
    \begin{equation}\label{Gk}
       \widehat{G}^{(k)} = W_k^\top G^{(k)} W_k= \begin{bmatrix}
           \widehat{G}^{(k)}_{p} &  \widehat{G}^{(k)\top}_{C} \\
           \widehat{G}^{(k)}_{C} &   \widehat{G}^{(k)}_{I} 
        \end{bmatrix};
     \end{equation}   
     \begin{equation}\label{Ck}
       \widehat{C}^{(k)} = W_k^\top C^{(k)} W_k= \begin{bmatrix}
          \widehat{C}^{(k)}_{p} & \widehat{C}^{(k)\top}_{C} \\
          \widehat{C}^{(k)}_{C} & \widehat{C}^{(k)}_{I} 
        \end{bmatrix}.
     \end{equation}
     Then we decompose matrix $\widehat{C}^{(k)}_{C}$:
     \begin{equation*}
       \widehat{C}^{(k)}_{C} = Q^{(k+1)}\begin{bmatrix}
                                     B^{(k+1)} \\
                                     0 
                                   \end{bmatrix},
     \end{equation*}
     where $B^{(k+1)}\in\mathbb{R}^{p_k\times p_{k+1}},\vspace{1em} p_{k+1}\leq p_k$, and construct orthogonal transformation $Q_{k+1}$:
     \begin{equation*}
       Q_{k+1} = \begin{bmatrix}
                   I_{p_k} &  \\
                    &  Q^{(k+1)}
                 \end{bmatrix}.
     \end{equation*}
     Then, through orthogonal transformation, we can obtain the conductance matrix and capacitance matrix of the $(k+1)$-th iteration:
     \begin{equation}\label{G_k+1}
         G^{(k+1)} \equiv Q^{(k+1)\top} \widehat{G}_{22}^{(k)}Q^{(k+1)};
     \end{equation}
     \begin{equation}\label{C_k+1}
         C^{(k+1)} \equiv Q^{(k+1)\top} \widehat{C}_{22}^{(k)}Q^{(k+1)}.
     \end{equation}
     After $m-1$ eliminations, the original system can be seen as the cascade of $m$ systems:
     \begin{align*}
        & \Sigma_1^{(1)}: \left\{\begin{aligned}
            & \left(\widehat{G}_{p}^{(1)} + s \widehat{C}_{p}^{(1)}\right) x^{(1)}_p = B^{(1)} u(s) + u^{(1)}(s), \\
            & y = B^{(1)\top} x^{(1)}_p, 
        \end{aligned}\right. \\
        & \Sigma_1^{(2)}: \left\{\begin{aligned}
            & \left(\widehat{G}_{p}^{(2)} + s \widehat{C}_{p}^{(2)}\right) x_p^{(2)} = -B^{(2)}u_2(s) + u^{(2)}(s), \\
            & y^{(2)} = -B^{(2)\top} x_p^{(2)}, 
        \end{aligned}\right. \\
        & \quad \quad \vdots \nonumber \\ 
        & \Sigma_1^{(m-1)}: \left\{\begin{aligned}
            & \left(\widehat{G}_{p}^{(m-1)} + s \widehat{C}_{p}^{(m-1)}\right) x_p^{(m-1)} =\\
            &-B^{(m-1)}u_{m-1}(s) + u^{(m-1)}(s), \\
            & y^{(m-1)} = -B^{(m-1)\top} x_p^{(m-1)}, 
        \end{aligned}\right. \\
        & \Sigma_{m-1}^{(2)}:\left\{\begin{aligned} &G^{(m)}x^{(m-1)}_I(s) + sC^{(m)}x^{(m-1)}_I(s) = -B_mu_m(s),\\
      &y(s) = -B_m^{\top} x^{(m-1)}_I(s),
      \end{aligned}\right. 
    \end{align*}
     where 
    \begin{equation*}
        u_k(s) = (s-s_{k-1})x_p^{(k-1)},\vspace{1em}2\leq k\leq m;
    \end{equation*}
    \begin{equation*}
        u^{(k)}(s) = (s-s_k)y^{(k+1)}(s),\vspace{1em}2\leq k\leq m-1.
    \end{equation*}
    
    Finally, we decouple the linear ports and the internal nodes in $\Sigma_{m-1}^{(2)}$ at expansion point $s_m$. Define matrix $A^{(m)}$:
    \begin{equation*}
        A^{(m)}\equiv G^{(m)} +s_mC^{(m)}=\begin{bmatrix}
            A_p^{(m)} & A_C^{(m)\top} \\
            A_C^{(m)} & A_I^{(m)}
        \end{bmatrix}.
    \end{equation*} The congruence transformation is 
    \begin{equation*}
        W_m = \begin{bmatrix}
            I_{p_m}&\\
            -\left(A_I^{(m)}\right)^{-1}A_C^{(m)}& I
        \end{bmatrix},
    \end{equation*} and we can obtain
    \begin{equation*}
        \widehat{G}^{(m)} = W_m^\top G^{(m)} W_m= \begin{bmatrix}
           \widehat{G}^{(m)}_{p} &  \widehat{G}^{(m)\top}_{C} \\
           \widehat{G}^{(m)}_{C} &  \widehat{G}^{(m)}_{I} 
        \end{bmatrix};
    \end{equation*}
    \begin{equation*}
        \widehat{C}^{(m)} = W_m^\top C^{(m)} W_m= \begin{bmatrix}
           \widehat{C}^{(m)}_{p} &  \widehat{C}^{(m)\top}_{C} \\
           \widehat{C}^{(m)}_{C} &  \widehat{C}^{(m)}_{I} 
        \end{bmatrix}.
    \end{equation*}  Perform the congruence transformation on $\Sigma_{m-1}^{(2)}$, and the original system can be seen as the cascade of $m$ small systems
    \begin{equation*}
      \Sigma_1^{(1)}:\left\{\begin{aligned} &\widehat{G}^{(1)}_{p}x^{(1)}_p(s) + s\widehat{C}^{(1)}_{p}x^{(1)}_p(s) = B^{(1)}u(s)+u^{(1)}(s)\\
      &y(s) = B^{(1)\top} x^{(1)}_p(s)
      \end{aligned}\right.;
    \end{equation*}
    \begin{equation*}
      \Sigma_2^{(1)}:\left\{\begin{aligned} &\widehat{G}^{(2)}_{p}x^{(2)}_p(s) + s\widehat{C}^{(2)}_{p}x^{(2)}_p(s) = -B^{(2)}u_2(s)+u^{(2)}(s)\\
      &y^{(2)}(s) = -B^{(2)\top} x^{(2)}_p(s)
      \end{aligned}\right.;
    \end{equation*}
    \begin{equation*}
      \vdots
    \end{equation*}
    \begin{equation*}
      \Sigma_m^{(1)}:\left\{\begin{aligned} &\widehat{G}^{(m)}_{p}x^{(m)}_p(s) + s\widehat{C}^{(m)}_{p}x^{(m)}_p(s) =\\& -B^{(m)}u_m(s)+u^{(m)}(s)\\
      &y^{(m)}(s) = -B^{(m)\top} x^{(m)}_p(s)
      \end{aligned}\right.;
    \end{equation*}
    and a large subsystem
     \begin{equation*}
      \Sigma_m^{(2)}:\left\{\begin{aligned} &\widehat{G}^{(m)}_{I}x^{(m)}_I(s) + s\widehat{C}^{(m)}_{I}x^{(m)}_I(s) = -C_C^{(m)}u_{m+1}(s)\\
      &y^{(m+1)}(s) = B^{(m)\top} x^{(m)}_p(s)
      \end{aligned}\right..
    \end{equation*}
\section{Proof of Multi-Frequency Points Matching}
\label{app:convergence}

After $m-1$ eliminations, the original system can be seen as the cascade of $m$ systems represented by \cref{Sigma_11,Sigma_21,Sigma_m1,Sigma_m2}.
    The transfer function of $\Sigma_2^{(m-1)}$ is denoted by 
    \begin{align*}
        H^{(m)}(s) = B_m^{\top} \left(G^{(m)} + sC^{(m)}\right)^{-1} B_m. 
    \end{align*}
    The transfer function of the system obtained by cascading $\Sigma_1^{(1)}, \ldots, \Sigma_1^{(m-1)}$ is denoted by $H^{(i)}(s),\vspace{1em}i=1,2,\cdots,m-1$. For system $\Sigma_1^{(m-1)}$, we have 
    \begin{align*}
        &\left(\widehat{G}_{p}^{(m-1)} + s \widehat{C}_{p}^{(m-1)}\right) x_p^{(m-1)} \\
        &=-B^{(m-1)}u_{m-1}(s) + u^{(m-1)}(s)\\
        &=-B^{(m-1)}u_{m-1}(s) + (s-s_{m-1})y^{(m)}\\
        &=-B^{(m-1)}u_{m-1}(s) + (s-s_{m-1})H^{(m)}(s)u_m(s)\\
        &= -B^{(m-1)}u_{m-1}(s) + (s-s_{m-1})^2H^{(m)}(s)x^{(m-1)}_p(s).
    \end{align*}
    Then we can obtain \cref{x_pm-1}. Therefore, we can represent the transfer function of system $\Sigma_1^{(m-1)}$ as
    \cref{H^m-1}.
    \begin{figure*}[h]
    \begin{equation}\label{x_pm-1}
        x_p^{(m-1)}(s)=-\left(\widehat{G}_{p}^{(m-1)} + s \widehat{C}_{p}^{(m-1)}-
        (s-s_{m-1})^2H^{(m)}(s)\right)^{-1}B^{(m-1)}u_{m-1}(s)
    \end{equation}
    \end{figure*}
    \begin{figure*}[h]
    \begin{equation}\label{H^m-1}
        H^{(m-1)}(s)=B^{(m-1)\top}\left(\widehat{G}_{p}^{(m-1)} + s \widehat{C}_{p}^{(m-1)}-
        (s-s_{m-1})^2H^{(m)}(s)\right)^{-1}B^{(m-1)}.
    \end{equation}
    \end{figure*}
    Then we have the recurrence relation 
    \begin{align*}
        &H^{(i)}(s) = B^{(i)\top} \left(\widehat{G}_{p}^{(i)} + s \widehat{C}_{p}^{(i)} - (s-s_{i})^2 H^{(i+1)}(s)\right)^{-1} B^{(i)}, \\ 
        & i = 1, \ldots, m-1. 
    \end{align*}

    For the reduced model , the $m$ subsystems are represented as \cref{Sigma_11,Sigma_21,Sigma_m1}.
    we have a similar recurrence relation 
    \begin{align*}
        & \widehat{H}^{(m)}(s) = B^{(m)\top} \left(\widehat{G}_{p}^{(m)} + s\widehat{C}_{p}^{(m)}\right)^{-1} B^{(m)}, \\
        & \widehat{H}^{(i)}(s) = B^{(i)\top} \left(\widehat{G}_{p}^{(i)} + s \widehat{C}_{p}^{(i)} - (s-s_i)^2 \hat{H}^{(i+1)}(s)\right)^{-1} B^{(i)}, \\
        &i = 1, \ldots, m-1. 
    \end{align*}

    For any given $s_0$, define 
    \begin{align*}
        \delta_i = \begin{cases}
            1, & \text{if } s_i = s_0,  \\
            0, & \text{otherwise}. 
        \end{cases}
    \end{align*}
    Beginning with $\widehat{H}^{(m)}(s) - H^{(m)}(s) = o\left((s-s_0)^{2\delta_{q}-1}\right)$, we can inductively obtain 
    \begin{align*}
        \widehat{H}^{(i)}(s) - H^{(i)}(s) = o((s-s_0)^{2\sum_{j=i}^q \delta_j - 1}), \quad i = 1, \ldots, m-1. 
    \end{align*}
    This is because 
    \begin{align*}
        & \widehat{H}^{(i)}(s) - H^{(i)}(s) \\
        = &\ B^{(i)\top} \left(\widehat{G}_{p}^{(i)} + s \widehat{C}_{p}^{(i)} - (s-s_i)^2 \widehat{H}^{(i+1)}(s)\right)^{-1} B^{(i)} -\\
        &B^{(i)\top} \left(\widehat{G}_{p}^{(i)} + s \widehat{C}_{p}^{(i)} - (s-s_{i})^2 H^{(i+1)}(s)\right)^{-1} B^{(i)} \\
        = &\ (s-s_i)^2 B^{(i)\top} \left(\widehat{G}_{p}^{(i)} + s \widehat{C}_{p}^{(i)} - (s-s_i)^2 \hat{H}^{(i+1)}(s)\right)^{-1}\\ &\left(H^{(i+1)}(s) - \hat{H}^{(i+1)}(s)\right)\\
        &\left(\widehat{G}_{p}^{(i)} + s \widehat{C}_{p}^{(i)} - (s-s_{i})^2 H^{(i+1)}(s)\right)^{-1} B^{(i)}.
    \end{align*}
    Finally, we have 
    \begin{align*}
        \widehat{H}^{(1)}(s) - H^{(1)}(s) &
        = o\left((s-s_0)^{2\sum_{j=1}^q \delta_j - 1}\right)\\
        &= o\left((s-s_0)^{2q(s_0)-1}\right), 
    \end{align*}
    which completes the proof.


\bibliographystyle{unsrt}
\bibliography{bibfile}

\end{document}